\documentclass{article}
\pdfoutput=1

\usepackage[numbers]{natbib}
     \usepackage[preprint]{neurips_2021}

\usepackage[utf8]{inputenc} 
\usepackage[T1]{fontenc}    
\usepackage{hyperref}       
\usepackage{url}            
\usepackage{booktabs}       
\usepackage{amsfonts}       
\usepackage{nicefrac}       
\usepackage{microtype}      
\usepackage{amsthm,bm}
\usepackage{mathtools}
\usepackage{dsfont}
\usepackage[pdftex]{graphicx}
\usepackage{wrapfig}
\usepackage{xcolor}
\usepackage{subfig}
\usepackage[linesnumbered, ruled, vlined]{algorithm2e}
\usepackage{cleveref}
\SetKwRepeat{Do}{do}{while}%
\DeclareMathOperator\supp{supp}
\newcommand{\indi}{\mathds{1}}

\DeclareMathOperator*{\esssup}{ess\,sup}

\title{Balancing detectability and performance of attacks on the control channel of Markov Decision Processes}

\author{%
  Alessio Russo\thanks{Corresponding author.}\qquad Alexandre Proutiere \\
  Division of Decision and Control Systems, EECS School\\
  KTH Royal Institute of Technology, Stockholm\\
  \texttt{\{alessior,alepro\}@kth.se}
}

\theoremstyle{definition}
\newtheorem{definition}{Definition}
\theoremstyle{lemma}
\newtheorem{lemma}{Lemma}
\theoremstyle{proposition}
\newtheorem{proposition}{Proposition}
\theoremstyle{theorem}
\newtheorem{theorem}{Theorem}
\theoremstyle{corollary}

\theoremstyle{plain}

\DeclareMathOperator*{\argmax}{arg\,max}
\DeclareMathOperator*{\argmin}{arg\,min}
\DeclareMathOperator*{\trace}{Tr}
\newcommand{\KL}{\textnormal{KL}}

\begin{document}

\maketitle

\begin{abstract}
We investigate the problem of designing optimal stealthy poisoning attacks on the control channel of Markov decision processes (MDPs). This research is motivated by the recent interest of the research community for adversarial and poisoning attacks applied to MDPs, and reinforcement learning (RL) methods. The policies resulting from these methods have been shown to be vulnerable to attacks perturbing the observations of the decision-maker. In such an attack, drawing inspiration from adversarial examples used in supervised learning, the amplitude of the adversarial perturbation is limited according to some norm, with the hope that this constraint will make the attack imperceptible. However, such constraints do not grant any level of undetectability and do not take into account the dynamic nature of the underlying Markov process. In this paper, we propose a new attack formulation, based on information-theoretical quantities, that considers the objective of minimizing the detectability of the attack as well as the performance of the controlled process. We analyze the trade-off between the efficiency of the attack and its detectability. We conclude with examples and numerical simulations illustrating this trade-off.
\end{abstract}

\section{Introduction}
The framework of Markov decision processes  (MDPs) has been successful in many applications of systems control \cite{puterman2014markov,sutton2018reinforcement}. Thanks to its simplicity, and generality, it is capable of modeling most of the dynamical processes. For unknown processes, reinforcement learning (RL) techniques have shown great potential in controlling unknown systems. As a matter of fact, during the last decade, we have witnessed an increased surge of interest in RL, where, by exploiting modern methods in Deep Learning \cite{lecun2015deep}, researchers were able to reach higher performance, sometimes surpassing human performance in games such as Go, Dota, and Atari games \cite{silver2016mastering,berner2019dota,mnih2013playing,mnih2015human}.
This increased interest  has made RL being applied more frequently in industrial applications, from temperature control in buildings \cite{chen2018optimal}, to health-care \cite{yu2019reinforcement}, financial trading \cite{deng2016deep} and more. However, as recently pointed out by Gartner and Microsoft \cite{burke2019gartner, kumar2020adversarial} , in the next years AI cyber-attacks will leverage data poisoning, or adversarial samples,  and \textit{only a small fraction of the companies have the right tools in place to secure their ML systems}.

Researchers have focused on attacks that poison the data used by RL to compute the control action. Simple types of attacks can be computed by means of the \textit{Fast Gradient Sign Method} (FGSM) \cite{goodfellow2014explaining,huang2017adversarial,pattanaik2018robust}, which computes a small perturbation of the data that minimizes some performance criterion. This attack has been shown to decrease the performance of RL agents when applied to observations of the state. Nonetheless, FGSM cannot compute optimal attacks. Instead, computing an optimal attack can be cast as an optimal control problem. This method of devising optimal attacks that poison the state observation has been shown in \cite{russo2021optimal,zhang2020robust}. Similarly, some attacks directly alter the action taken by the agent, instead of the state \cite{tan2020robustifying,tessler2019action}  However, an issue of this body of work is that detectability is measured in terms of a distance metric, that is usually taken to be the $\ell_2$-norm, or the $\ell_\infty$-norm, and the attack amplitude is constrained according to this metric. Unfortunately, this type of constraint does not take into account the dynamic nature of the underlying MDP, and therefore it is just an \textit{approximated} way to deal with detectability.

To this aim, we propose a new attack formulation based on the idea that the adversary wants to minimize \textit{detectability} as well as \textit{performance} of the agent. The problem of detectability can be framed as a hypothesis testing problem, and we motivate a new attack criterion based on the theory of \textit{quickest change detection} \cite{lai1998information,basseville1993detection,tartakovsky2014sequential}. We focus our attention on attacks on the control channel of an MDP, and frame the detectability problem as a quickest change detection problem. We provide a new definition of attack detectability, and show how to compute attacks that minimize this detectability metric as well as performance. We conclude with examples and numerical simulations illustrating this trade-off. 

\textit{Structure of the paper.} In \cref{sec:related_work}, we present the related work, and introduce the framework of Markov decision processes. In \cref{sec:optimal_attack}, we formulate the problem of optimally attacking the control channel of an MDP.  We conclude with simulations in \cref{sec:simulations}. 
\section{Related work and preliminaries}\label{sec:related_work}
\textbf{Adversarial attacks in machine learning.} Only recently  researchers have started to address the problem of adversarial attacks on machine learning methods. This interest has originally sparked from an analysis \cite{szegedy2013intriguing,goodfellow2014explaining} that showed how deep learning models are affected by the \textit{adversarial example} phenomenon. An adversarial example is a type of perturbation that carefully alters the input data of a machine learning model with the goal of reducing the performance of the model. Technically, one aims at finding a small perturbation that if added to the data can significantly decrease the model's performance. This is usually done using an attack that relies on the gradient of the loss function of the model (check the FGSM attack for an example \cite{goodfellow2014explaining}). Many other attacks have been developed using this principle, and most of the defenses use \textit{adversarial training} (i.e. the model is robustified by training on perturbed data), distilled policies or robust neural networks \cite{kurakin2016adversarial,madry2017towards,carlini2017towards,yuan2019adversarial,papernot2016distillation}. 

\textbf{Adversarial attacks in reinforcement learning.} Researchers have started to also analyze the problem for reinforcement learning agents (one can refer to \cite{chen2019adversarial} for a brief summary). Initially, the focus has been on FGSM-like attacks on the observations of a Markov process \cite{huang2017adversarial,pattanaik2018robust,behzadan2017vulnerability,lin2017tactics}, or attacks that directly affect the state of the system. The latter type of attack usually studies an adversary that can directly affect the system, and the goal is to find a policy that is robust against the worst adversary by solving a minimax game where also the adversary is trying to control the MDP \cite{morimoto2005robust,pinto2017robust}. 
However, attacks on state observations perturb only the state measurement, but not the actual state of the system. To craft this attack using FGSM-like methods one usually uses the $Q$-value of a policy, since there is no loss function to consider in RL. Nonetheless, using the $Q$-value of a policy leads to sub-par attacks. This is due to the fact that it is equivalent to find a perturbation that minimizes the instantaneous reward, whilst optimal attacks should minimize the entire trajectory of rewards. Optimal attacks on the observations of a Markov process can be found by solving an adversarial MDP, as shown in \cite{russo2021optimal,zhang2020robust}.
Attacks on the observations lead to a partially observable model (a.k.a. POMDP), and therefore it is hard to find robust policies.  Some of the defense mechanisms rely on the concept of adversarial training, policy distillation, the usage of history of data or the use of recurrent layers in neural networks \cite{chen2019adversarial, russo2021optimal}.

Similarly to attacks that directly affect the state of the systems, there are attacks on the control channel of an MDP, i.e., attacks that alter the action chosen by the victim. This is in contrast with previous studies on robust MDPs, where the transition dynamics still depend on the action chosen by the victim. In this case, the adversary sits in between the victim's policy and the MDP. In  \cite{tessler2019action} they analyze the case where the action is randomly perturbed by an adversary, and analyze how to robustify the agent's policy against these perturbations. To find a robust policy they frame the problem as a max-min game, but do not consider the problem of a stealthy attack. In \cite{tan2020robustifying} the authors consider an FGSM-like attack on the control channel, and propose adversarial training as a way to robustify the policy. In contrast,  in \cite{lee2020spatiotemporally}  to compute an attack the authors propose to solve an optimization problem that minimizes the cumulative reward over a finite horizon, subject to budget constraints. To do so, they solve the optimization problem using a projected gradient descent method, and therefore can be considered an FGSM-like method.

\textbf{Markov decision process (MDP).}
An MDP $M$ is a controlled Markov chain, described by a tuple $M=(S, A, P, r, p_0)$, where $S$ and $A$ are the state and action spaces, respectively. $P: S\times A \to \Delta(S)$ denotes the conditional state transition probability distributions ($\Delta(S)$ denote the set of distributions over $S$), i.e., $P(s'|s,a)$ is the probability to move from state $s$ to state $s'$ given that action $a$ is selected. We also write $P(s,a)$ to denote the distribution over the next state given $(s,a)$. Finally, $p_0$ is the initial distribution of the state and $r: S\times A\to [0, R^\star]$ is the reward function, with $R^\star>0$. A (randomized) control policy $\pi: S\to \Delta(A)$ determines the selected actions, and $\pi(a|s)$ denotes the probability of choosing $a$ in state $s$ under $\pi$. Here we focus on ergodic MDPs, where   any policy $\pi$ generates a positive recurrent Markov chain. The discounted value of a policy $\pi$ is defined as $V_\gamma^\pi(s) =  \mathbb{E}^\pi\left[ \sum_{t\geq0}\gamma^t r(s_t, a_t)|s_0=s\right]$ (here $a_t$ is distributed according to $\pi(\cdot|s_t)$) for any initial state $s$, and discount factor $\gamma\in(0,1)$, whilst its ergodic reward (or average reward) is defined as $h^\pi = \lim_{N\to\infty} \mathbb{E}^\pi\left[\frac{1}{N+1}\sum_{t\geq0}^N r(s_t, a_t)\right]$.

\section{Optimal attacks on the control action}\label{sec:optimal_attack}
In this section, we first model the attack problem as a sequential decision-making problem. Then, we discuss two approaches to make the attack stealthy. The first approach limits the set of actions available to the adversary. The second one uses the definition of \textit{information rate} to define stealthy attacks. Lastly, we conclude with the formulation of optimal stealthy attacks.
\begin{wrapfigure}{HR}{0.35\textwidth}
\centering
\includegraphics[width=0.35\textwidth]{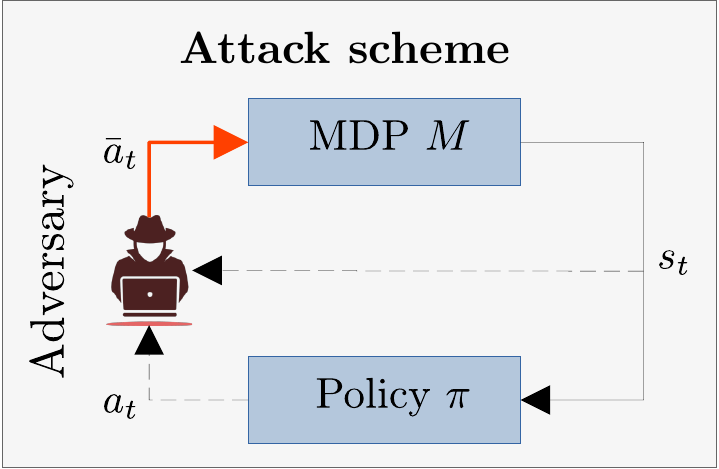}
\caption{\label{fig:scheme}Adversarial attacks on the control channel.}
\end{wrapfigure}
\subsection{The attack MDP}
\textbf{Problem description.} Here we describe the problem setting and how an \textit{adversarial agent} attacks the control channel of a decision-maker, which we call \textit{victim} in the following. First, we assume the \textit{victim} uses a stationary Markov policy $\pi$, not necessarily deterministic, with the goal of maximizing the total collected reward. We then assume that \textit{the adversary is capable of measuring the state} $s_t$, and \textit{can manipulate the action taken at the input channel of the MDP}. This condition implies that the adversary can change the action $a_t$ taken by the victim, and we denote the poisoned action by $\bar a_t$. On the other hand, \textit{the victim is not able to measure the perturbed action $\bar a_t$ chosen by the adversary}. 

Additionally, we assume \textit{the reward function $r$ is chosen by the victim} (and we assume it is known by the adversary), computed according to the state-action pair $(s_t,a_t)$ in round $t$. This is a classical assumption in control theory, where the reward is built according to the state measurements of the system. This is in stark contrast with previous studies \cite{tessler2019action}, where they considered a reward that depends on the perturbed action $\bar a_t$, and not the original one $a_t$. A consequence is that \textit{ it is not possible to use the reward function as a way to detect the presence of anomalies}, thus making the problem harder to solve. Nonetheless, this is not a necessary assumption, and one can relax it to take into account also the reward signal as explained later in the text.

\textbf{Attack MDP.} The goal of the adversary is to minimize the performance of the victim. Under these assumptions, the problem of finding an optimal attack can be cast to that of solving a Markov Decision Processes. In fact, note that for a stationary Markov policy $\pi$ the system $M\circ\pi$ can be modeled as an MDP.  As a consequence, we can define an \textit{attack MDP} $\bar M$ that the adversary wishes to control. Formally, the MDP the adversary wishes to solve is $\bar M = (S\times A, A, P^\pi, \bar r)$, where  $P^\pi(s',a'|s,a,\bar a)= \pi(a'|s') P(s'|s, \bar a)$ $\forall (s,\bar a, a, s', a')$ . The adversarial reward $\bar r: S\times A\times A \to \mathcal{R}$ (with $\mathcal{R}$ being a compact closed subset of $\mathbb{R}$) is chosen by the adversary, and can be simply put to $\bar r(s,a,\bar a) = -r(s,a)$ to obtain the classical zero-sum game formulation between two agents. A consequence of this formulation is that the adversary only needs to consider stationary Markov policies to optimally solve the problem. For an attack policy $\phi:S\times A\to \Delta(A)$ we denote the overall policy of the system by $\phi\circ\pi : S\to \Delta(A)$. Finally, we  denote respectively by $\bar V_{\bar \gamma}^{\phi\circ \pi}(s,a)=\mathbb{E}^{\phi\circ\pi}[\sum_{t\geq 0}\bar\gamma^t  \bar r(s_t, a_t, \bar a _t)| s_0=s, a_0=a]$ the discounted value of the adversarial policy for a discount $\bar \gamma$. 
Similarly, for any attack $\phi$ we denote the discounted value 
of the attacked policy $\pi$ by $V_\gamma^{\phi\circ\pi}(s) = \mathbb{E}^{\phi\circ\pi}\left[\sum_{t\geq 0} \gamma^t r(s_t,a_t)\Big | s_0=s\right]$ 
, where $s_{t+1} \sim P^\phi(\cdot|s_t,a_t)$ and $P^\phi(s'|s_t,a_t) = \mathbb{E}_{\bar a \sim \phi(s_t,a_t)}[P(s'|s_t, \bar a)]$.  Given these premises, for any attack $\phi$ we can find an upper bound of the regret of the victim, similar to the one in \cite{zhang2020robust}.
\begin{proposition}\label{proposition:bound_value}
	Consider an MDP with bounded reward $|r(s,a)|\leq R^\star$. The difference of the discounted value of the policy $\pi$, and the policy under attack $\phi\circ\pi$, is upper bounded as follows
	\begin{equation}
		\|V_\gamma^\pi - V_\gamma^{\phi\circ\pi}\|_\infty \leq  \alpha\max_{s,a,\bar a}\|P(\cdot|s, a) - P(\cdot|s,\bar a)\|_{TV},
	\end{equation}
	where $\|P(\cdot|s, a) - P(\cdot|s,\bar a)\|_{TV}$ is the total variation distance between $P(\cdot|s, a)$ and  $P(\cdot|s,\bar a)$, and $\alpha=2 \gamma R^\star/(1-\gamma)^2$ is a constant term\footnote{The reader can find all the proofs  in the appendix.}.
\end{proposition}
This inequality differs from the one   in \cite{russo2021optimal,zhang2020robust}, where the upper bound also depends on $\pi$. Here, instead of having a total variation on the policy, we have that the bound depends solely on the transition density. Even though the bound may seem loose, the fact that we do not get a stronger dependency on $\pi$, as in \cite{russo2021optimal,zhang2020robust}, seems to suggest that the regret problem mostly depends on the underlying MDP, than the chosen policy. Then, \textit{it may not be always possible to find a robust policy}. As a consequence, attack detection may be preferable. Since previous studies have not considered detectability from a statistical point of view, this leads us to study the problem of attack detectability.
\subsection{Detectability constraints on the attack}
In this section we study the problem of making an attack less detectable. We first consider simple constrained attacks, and argue how these attacks do not provide any stealthiness guarantee, and then proceed to study stealthy attack from a statistical point of view.

\textbf{Optimal constrained attack.}
 Stealthiness in literature has usually been defined as \textit{how close} is the perturbed signal to the real signal (using a distance function $d$, or a norm). This assumption  carries out the idea that somehow the victim is checking the goodness of the measured data.
\begin{definition}[Constrained stealthy attack]\label{def:stealthy_attack_constrianed}
Let $d:A\times A \to [0,\infty)$ be a distance function and let $\varepsilon\geq 0$. We define an attack policy $\phi$ to be $(d,\varepsilon)$-constrained if, for any $(s,a)\in S\times A$, the support of $\phi$ in $(s,a)$ is $\bar A_a(\varepsilon;d) = \{\bar a \in A: d(\bar a, a) \leq \varepsilon \}$.
\end{definition}
 This notion of stealthiness can be easily adopted to compute an optimal constrained attack. For the discounted value (similarly also for other criterion) the attack is defined to be the optimizer of the following  problem: for any  $(s,a)\in S\times A,\bar\gamma \in (0,1)$
\begin{equation}\label{eq:problem_constrained_attack}
\max_{\phi\in \Phi(\varepsilon;d)} \bar V_{\bar \gamma}^{\phi\circ\pi}(s,a) \hbox{ s.t. }  s_{t+1}\sim P(\cdot| s_t, \bar{a}_t),
\end{equation}
where $\Phi(\varepsilon;d)=\{\phi \in \Delta(A)^{S\times A} : \supp(\phi(s,a)) \subseteq \bar A_a(\varepsilon;d),\forall (s,a) \in S\times A\}$, for $\varepsilon>0$ and a metric $d$. The previous optimization problem  results in an optimal policy $\phi^\star$ that is deterministic, stationary and Markovian The problem can be easily solved both in the case the adversary knows the model, i.e., knows $(P,\pi)$, and also in the case where the model is not known. In the former case, that we denote also as \textit{white-box}, the adversary can solve the MDP by means of Value Iteration or Policy Iteration. In the latter case, that we denote as \textit{black-box} case, it is possible to use RL techniques, such as Q-learning or policy-gradient based methods, to compute an optimal attack policy $\phi^\star$. Consequently, we omit to describe an algorithm that solves \cref{eq:problem_constrained_attack}. 

However, we argue that constrained attacks are in general not stealthy. Constraining the amplitude of an attack does not necessarily imply a decrease in detectability for the following two reasons: (1) it depends on what kind of detection method the victim is using; (2) it does not consider the dynamics of the underlying process. Moreover, this notion of stealthiness tends to be useful as long as the victim can compare the measured signal with some reference signal (where the comparison is done using the metric $d$). However, this may not be always the case, or the adversary may not know what is the metric $d$. These arguments lead us to consider a different concept of stealthiness, based on statistical detectability.

\textbf{Information-theoretical stealthiness} 
We introduce a different notion of stealthiness based on information theoretical quantities. Attack detection in MDPs can be framed as a minimax quickest change detection problem (QCD), which is the problem of detecting a change point at which the statistical properties of the stochastic process undergo a change \cite{veeravalli2014quickest,tartakovsky2014sequential,basseville1993detection}. We consider minimax QCD  since we do not know when the adversary attacks the system, and therefore the distribution of the change point is unknown. 

The setup is as follows: we denote by $\nu$ the unknown change time at which the attack $\phi$ is applied to the MDP controlled by $\pi$.  Assume the reward signal $r$ is not provided by the MDP, but constructed directly from the observation of $(s,a)$ (if that is not the case, we can easily change the following argument by considering $(r,s,a)$ instead\footnote{In this case, for the problem to be well-posed, we require the rewards to be randomized; see also the appendix for a proper formulation of this particular case.}). Consequently, we consider  a sequence of non-i.i.d. observations $\{(s_t, a_t)\}_{t\geq 0}$, and  assume the conditional density of $(s_t,a_t)$ given the previous measurement is $P(\cdot|s_{t-1}, a_{t-1})$ for $t<\nu$, and $P^\phi(\cdot|s_{t-1}, a_{t-1})$ otherwise. 

The victim needs to decide in each round if she is under attack. Consequently, her decision takes the form of a stopping rule $T$ (i.e., a detection rule) with respect to the filtration $(\mathcal{F}_t)_t = (\sigma(s_0,a_0,\dots, s_t,a_t))_t$.  For this setup, a common criterion of performance due to Lorden and Pollak \cite{lorden1971procedures,pollak1985optimal} is the worst case expected delay
$\overline{\mathbb{E}}(T)=\sup_{\nu \geq 1}\esssup \mathbb{E}_\nu[(T-\nu)^+|\mathcal{F}_{\nu-1}]$, where the expectation $\mathbb{E}_\nu[\cdot]$ is taken with respect to the underlying probability measure when the change happens at time $t=\nu$. The goal of the victim  is to minimize $\overline{\mathbb{E}}(T)$ over all stopping rule satisfying $\mathbb{E}_\infty[T] \geq \overline T$, for some $\overline T>0$, when $\nu=\infty$. The constraint, in simple words, lower bounds the expected duration to false alarm. 

Having described the detection problem, we know that the following asymptotic lower bound \cite{lai1998information}  holds under some assumptions that are satisfied by ergodic Markov chains: \begin{equation}\lim_{\overline T \to\infty} \inf_{T: \mathbb{E}_\infty[T]\geq \overline T} \overline{\mathbb{E}}[T]/\ln(\overline{T}) \geq I^{-1},\end{equation}
for a constant $I>0$. 
In brief, this lower bound characterizes the sample complexity of detecting  a change in the model.  
Specifically, $I$  measures the average rate of information that the victim can use to discriminate between two hypotheses. Because of this lower bound, \textit{the rate $I$ directly affects the sample complexity of the detection problem, which increases as $I$ decreases}. Moreover. this lower bound is matched, asymptotically, by the CUSUM rule.

Therefore, the idea of the adversarial agent is to choose a policy $\phi$ that minimizes $I$. Consequently,  detectability decreases since the detection delay of the victim increases. To proceed with this idea, we first observe that there is a clear link between $I$ and the log-likelihood ratio (LLR). Let the LLR between the two models be
 $z_\phi(s,a,s',a')\coloneqq  \ln\frac{P^\phi(s',a'|s,a)}{P(s',a'|s,a)}$. Consequently, $z_\phi$ is equal to
 \begin{equation}
 	z_\phi(s,a,s',a') = \ln\frac{\mathbb{E}_{\bar a \sim \phi(s,a)}\left[\pi(a'|s')P(s'|s,\bar a)\right]}{\pi(a'|s') P(s'|s,a)} =\ln\frac{P^\phi(s'|s,a)}{P(s'|s,a)}.
\end{equation}
Note that $z_\phi$ does not depend on $a'$ and $\pi$, but solely on $\phi$ and $(s,a,s')$.  Therefore we simply write $z_\phi(s,a,s')$ in the following. 
Now, we  exploit the idea that for ergodic models the expected value of $z_\phi$ for $t\geq \nu$ converges to $I$, which, in this case, depends also on $(\pi,\phi)$, and we  denote it by $I(\pi,\phi)$. Let $\mathcal{C}(\pi) = \{(s,a): \pi(a|s)>0\}$ be the set of possible state-action pairs, and  assume that $\phi$ satisfies $P^\phi(s,a) \ll P(s,a), \forall (s,a)\in \mathcal{C}(\pi)$. Then, it is possible to prove that for ergodic MDPs the quantity $n^{-1}\sum_{t=\nu}^{\nu+n} z_\phi(s_t,a_t,s_{t+1})$ converges to  $I(\pi,\phi)$  as $n\to\infty$ (see \cite{lai1998information}). This argument motivates the following definition of stealthy attacks.

\begin{definition}[Information-theoretical stealthy attack]\label{def:stealthy_attack_it}
For $\varepsilon\geq 0$ we define an attack policy $\phi$ to be $\varepsilon$-stealthy if $ I(\pi, \phi)\leq \varepsilon$, where $ I(\pi,\phi)$ is the information rate number
\begin{align*}
	I(\pi,\phi) &= \mathbb{E}_{s\sim \mu^{\phi \circ \pi}, a\sim \pi(\cdot|s)} \left[ \KL(P^\phi(s,a), P(s,a))\right],
\end{align*}
 $\KL(\cdot,\cdot)$ is the KL-divergence, and $\mu^{\phi \circ \pi}$ is the on-policy distribution induced by $\phi$ and $\pi$.
\end{definition}

\subsection{Optimal information-theoretical stealthy attacks}\label{subsec:optimal_it_stealthy_attack}

Intuitively, \cref{def:stealthy_attack_it}  better captures the idea of a stealthy attack than \cref{def:stealthy_attack_constrianed} (note that the two ideas are not mutually exclusive, and can be  combined together).  The smaller $I(\pi,\phi)$ is, the harder it is for the victim to distinguish and decide between the hypothesis of being under attack or not. 
Based on \cref{def:stealthy_attack_it}, we can design an attack that minimizes performance as well as statistical detectability: for any  $(s,a)\in S\times A,\bar\gamma \in (0,1), \varepsilon\geq0$
\begin{equation}\label{eq:stealthy_attack_original}
\max_{\phi\in \Phi(P,\pi)} \bar V_{\bar \gamma}^{\phi\circ\pi}(s,a) \hbox{ s.t. }  I(\pi,\phi) \leq \varepsilon \hbox{ and } s_{t+1}\sim P(\cdot| s_t, \bar a_t),
\end{equation}
where  $\Phi(P,\pi)=\{\phi \in \Delta(A)^{S\times A}:  P^\phi(s,a) \ll P(s,a), \forall (s,a) \in \mathcal{C}(\pi)\}$. Unfortunately this attack formulation can not be easily solved. A reason is that $I(\pi,\phi)$ is formulated in terms of the on-policy distribution, which makes \cref{eq:stealthy_attack_original} hard to solve in presence of a discount factor. Moreover, even in case the adversary considers an ergodic reward criterion $\lim_{N\to\infty}\frac{1}{N-\nu}\mathbb{E}\left[\sum_{t=\nu}^N \bar r(s_t, a_t, \bar a_t)\right]$, instead of  $\bar V_{\bar \gamma}^{\phi\circ\pi}$, the optimization problem is still non-trivial. This is due to the dependency on $\phi$ of the KL-divergence term in $I$, which makes, in general, the maximization problem  convex in the state-action distribution induced by the policy (i.e., $\phi\circ\pi$) (therefore with multiple solutions attained at the boundaries of the feasible set; see also the appendix for a discussion).

Instead of solving \cref{eq:stealthy_attack_original}, we make use of the following observations: (1) first, we find an upper bound on $I(\pi,\phi)$ that permits us to remove the dependency on $\phi$ from the KL-divergence term;   (2) secondly, we observe that we can use a discounted criterion in place of the ergodic criterion in \cref{def:stealthy_attack_it} as long as the discount factor $\bar \gamma$ is close to $1$.

\textbf{Upper bounding $I$.} The following lemma uses the log-sum inequality to upper bound $I(\pi,\phi)$.
\begin{lemma}\label{lemma:upper_bound_I}
	Assume that $\phi$ satisfies $P(s,\bar a) \ll P(s,a)$ for every $(s,a,\bar a)\in \mathcal{C}(\pi,\phi),$ where  $\mathcal{C}(\pi,\phi) =\{(s,a,\bar a): P(s,\bar a)\ll P(s,a) \wedge \pi(a|s) \phi(\bar a|s,a)>0\}$.  Then, the information value $I(\pi,\phi)$ can be upper bounded by $\bar I(\pi,\phi)$ as follows
\begin{equation}
	I(\pi,\phi) \leq  \mathbb{E}_{s\sim \mu^{\phi \circ \pi}, a\sim \pi(\cdot|s),\bar a\sim\phi(\cdot|s,a)}\left[\KL(P(s,\bar a), P(s,a)) \right] \eqqcolon \bar I(\pi,\phi).
\end{equation}
\end{lemma}
Observe that the absolute continuity assumption $P(s,\bar a) \ll P(s,a)$ can be easily verified in those systems whose state is affected by some form of process noise (like  exogenous stochastic disturbances of the state).  We now consider the second simplification.

\textbf{Discounted information rate.} The second simplification permits us to consider a discounted version of \cref{def:stealthy_attack_it}. This change allows the use of discounted methods, which in turn permits to consider also the transient trajectory of the system in the information rate. Note, though, that this change is unnecessary if the adversary aims to maximize her ergodic reward, instead of $\bar V_{\bar \gamma}^{\phi\circ\pi}$. Since the result holds also for other type of problems, we state it in a general form. The key observation is that  for  a large discount factor we can approximate the gain of a chain with its discounted value.
\begin{proposition}\label{proposition:information_lb_discounted}
 Consider a Markov chain  $\{x_t\}_{t}$  over a finite space $\mathcal{X}$. Consider two transition functions $P_1, P_0$ over $\mathcal{X}$. Assume that for $t<\nu$ the distribution of $x_t$ given $x_{t-1}$ is $P_0(\cdot|x_{t-1})$, while for $t\geq \nu$ is $P_1(\cdot|x_{t-1})$, with $P_1(x)\ll P_0(x)$ for all $x$. Assume the chain is positive recurrent under $P_1$, with stationary measure $\mu$. Let $\gamma \in(0,1)$ and define $I_\gamma(x) = \mathbb{E}\left[\sum_{t\geq \nu}\gamma^t (1-\gamma)\KL(P_1(x_t), P_0(x_t)) |x_{\nu}=x\right]$. Then, for all $x$ we have
\begin{equation}
I= \mathbb{E}_{x\sim \mu}[\KL(P_1(x), P_0(x))] = \lim_{\gamma \to 1} I_\gamma(x).
\end{equation}
\end{proposition}
Additionally, we also have the following proposition that bounds the error we make by considering $I_\gamma$ instead of the information term $I$. This bound can also be generalized to general state-action spaces by considering the Laurent decomposition shown in \cite{jasso2009blackwell}, Theorem 3.1.
\begin{proposition}\label{proposition:discounted_information_term_error}
	Suppose the chain $\{x_t\}_{t}$ is aperiodic and uniformly ergodic under $P_1$, that is $\sup_{x\in \mathcal{X}}\|P_1^t(x) - \mu\|_{TV} \leq L \theta^t$ for some $L>0$ and $\theta\in (0,1)$.
	Let $ D^\star = \max_{x} \KL(P_1(x), P_0(x))$ and $\gamma_0 = 1/(1+(1-\theta)L)$. Then, for $\gamma \in (\gamma_0, 1)$ we have that
	\begin{equation}
	\sup_x |I_\gamma(x) - I| \leq \frac{(1-\gamma) LD^\star}{\gamma(1-\theta) - (1-\gamma)L},
	\end{equation}
	which converges to $0$ as $\gamma\to 1$.
\end{proposition}

\textbf{Approximated stealthy attack.} Combining the two ideas, for $\bar \gamma$ sufficiently close to $1$ we define 
\[
\bar I_{\bar \gamma}(s,a) = \mathbb{E}^{\phi\circ\pi}\left[\sum_{t\geq 0}{\bar \gamma}^t (1-\bar \gamma)\KL(P(s_t, \bar a_t), P(s_t, a_t)) \Big|s_{0}=s, a_0=a\right],
\]
where, according to \cref{proposition:information_lb_discounted}, we have $\lim_{\bar \gamma \to 1} \bar I_{\bar \gamma}(s,a) = \bar I(\pi, \phi)$ for any $(s,a)$, with $I(\pi,\phi)\leq \bar I(\pi,\phi)$. By rewriting  $\bar I_{\bar \gamma}$ in terms of the discounted state distribution $\mu_{\bar \gamma}^{\phi\circ\pi}$ induced by $\phi\circ \pi$, we  get  the following formulation of an optimal stealthy attack.

\begin{proposition}\label{prop:stealthy_randomized_attack}
An optimal attack $\phi^\star$ is $\varepsilon$-stealthy, according to $\bar I_{\bar \gamma} $, if it is an optimizer of the following problem: for $\bar \gamma \in (0,1)$, for any $(s,a)\in S\times A$
\begin{equation}\label{eq:simplified_stealthy_attack_problem}
\max_{\phi \in \Phi'(P,\pi)}  \bar V_{\bar \gamma}^{\phi\circ\pi}(s,a),\textrm{ s.t. } \mathbb{E}_{s\sim \mu_{\bar \gamma}^{\phi\circ\pi},a\sim\pi(\cdot|s),\bar a \sim \phi(\cdot|s,a)}[\KL(P(s,\bar a), P(s,a))] \leq \varepsilon
\end{equation}
where $\mu_{\bar \gamma}^{\phi\circ\pi}$ is the discounted state distribution induced by $\phi\circ\pi$, and $\Phi'(P,\pi)=\{\phi \in \Delta(A)^{S\times A}:  P(s, \bar a) \ll P(s,a), \forall (s,a, \bar a) \in \mathcal{C}(\pi,\phi)\}$.  The problem in \cref{eq:simplified_stealthy_attack_problem} admits an optimal policy that is stationary, Markov and randomized.
\end{proposition}
The problem in \cref{prop:stealthy_randomized_attack}  can be cast as a linear program in terms of the discounted state-action discounted  induced by $\phi\circ \pi$ (see the appendix for more details).  Additionally, we can also determine a very useful metric, that is the hardness of detecting an attack on an MDP $M$ controlled by $\pi$. 
\begin{proposition}\label{prop:mnimum_achievable_infrate}
Consider any attack $\phi$ that results in an ergodic reward of the victim to be at-most $\rho$, with $\rho\leq \mathbb{E}_{s\sim \mu^{\pi},a\sim\pi(s)}[r(s,a)]$. Then, the minimum achievable information rate $\bar I(\pi,\phi)$ can be computed by solving the following linear program
\begin{equation}\label{eq:inf_rate_problem}
\min_{\phi \in \Phi'(P,\pi)} \bar I(\pi,\phi) \textrm{ s.t. } \mathbb{E}_{s\sim \mu^{\phi\circ\pi},a\sim\pi(\cdot|s)}[r(s,a)] \leq \rho
\end{equation}
\end{proposition}
The problem in \cref{eq:inf_rate_problem} computes the least detectable attack in the set of attacks that make the ergodic reward of the victim to be at-most $\rho$.
The result  can be used to measure the detection hardness as a function of $\rho$, and can help the user compare how different policies $\pi$ affect detectability. Note that the optimization problem considers the ergodic reward. In case it is necessary to use a discounted reward, it is possible to prove that the problem becomes non-convex, unless one replaces $\bar I(\pi,\phi)$ with $\bar I_\gamma(\pi,\phi)$ (see the appendix for a formulation that used a discounted reward).
 
\textbf{Reinforcement learning approach.} A consequence of \cref{prop:stealthy_randomized_attack} is that deterministic policies are in general suboptimal. If  RL techniques are used,  then it is necessary to use a stochastic actor in order to find an optimal solution, otherwise stealthiness may not be guaranteed.
The problem in \cref{prop:stealthy_randomized_attack} is already formulated as a constrained MDP optimization problem \cite{altman1999constrained}, and therefore can be solved using constrained-policy optimization techniques, such as CPO \cite{achiam2017constrained} or PDO \cite{chow2017risk}, where we set the constraint to be $C(\pi,\phi)=\mathbb{E}^{\phi\circ\pi}[\sum_{t\geq 0} \bar \gamma^t (1-\bar \gamma) \KL(P(s_t,\bar a_t), P(s_t,a_t))] \leq \varepsilon$.

 Alternatively, instead of using constrained-policy optimization techniques,  it is still possible to use standard RL algorithms, like SAC \cite{haarnoja2018soft} or PPO \cite{schulman2017proximal}, by simply considering an augmented reward term $\bar r_\beta$ that penalizes the KL-divergence with a penalty factor $\beta>0$: 
$\bar r_\beta(s_t,a_t,\bar a_t)  \coloneqq \bar r(s_t,a_t,\bar a_t)  -\beta (1-\bar \gamma) z(s_t,a_t,\bar a_t,s_{t+1})$, where $z(s,a,\bar a,s') = \ln \frac{P(s'|s,\bar a)}{P(s'|s,a)}.$
   If the likelihood ratio is not known, it is possible to use a two-time scale stochastic approximation algorithm \cite{borkar2009stochastic} to both learn the policy $\phi$ and the likelihood ratio $z$, where $\phi$ is learnt at a slower pace than $z$.

\section{Examples and numerical results}\label{sec:simulations}
We now consider two significant examples: the inventory control problem, and the control of linear dynamical systems\footnote{Link to the code: \url{github.com/rssalessio/optimal-attack-control-channel-mdp}.}. We use these two examples to  demonstrate the possibility of crafting stealthy attacks capable of minimizing performance. For the inventory control problem we evaluate the efficiency of the various attack models, i.e., the constrained attack in \cref{eq:problem_constrained_attack}, the optimal stealthy randomized attack in \cref{eq:simplified_stealthy_attack_problem}, and a deterministic attack computed using the reward  $\bar r_\beta(s_t,a_t,\bar a_t)$, with penalty factor $\beta>0$, defined in the previous section. Lastly,  we study how to craft stealthy attacks that minimize the average reward of a linear system.
\subsection{The inventory control problem}
\textbf{Description.} The inventory control problem is a widely known problem in literature (see, e.g., \cite{szepesvari2010algorithms}), and concerns the problem of managing an inventory of fixed maximum size $N$ in  face of uncertain demand. For brevity, the details of this  problem can be found in the appendix.

\textbf{Attack evaluation.} In the left plot of \cref{fig:inventory_attack_evaluation} are shown results for the various attacks as function of their respective parameters. For the constrained attack in \cref{eq:problem_constrained_attack} we used a distance function $d(s,s')=|s-s'|$, and constrained the set of available actions to $\mathcal{C}(\pi^\star)$ in order to avoid that the information rate goes to infinity. We evaluated the best attack policy $\phi^\star$ for each problem against the best policy $\pi^\star(s) \in \argmax_\pi V_\gamma^\pi(s)$, and we plotted the normalized average discounted reward of the victim's policy. Since the reward depends also on the next state, $r(s,a)$ is an expectation that takes into account the distribution of the next state.
We see that the optimal randomized attack according to \cref{eq:simplified_stealthy_attack_problem} (orange curve) achieves larger performance decrease as well as lower detectability, due to a lower value of $I$. On the other hand, the deterministic attack found using the reward $\bar r_\beta$ shows a discrete behavior: for values of $\beta$ approximately lower than $12$ we have $I\approx 0.56$, and $I=0$ otherwise. The fact that for decreasing $\beta$ we see a decreasing reward, but constant $I$, is due to the fact that the attack is decreasing the detectability during the transient, and not at stationarity, since we are using $\bar I_{\gamma}$, a discounted version of the information rate. 

\textbf{Attack detection}. We  evaluated attack detectability using the optimal CUSUM detector $T_c=\inf\{t: c_t \geq c\}$, with $c_t=(\max_{1\leq k\leq t} \sum_{n=k}^t z_\phi(s_n, a_n, s_{n+1}))^+$, and a Generalized Likelihood Ratio (GLR) rule $T_g =\inf\{t:  g_t \geq c\}$, with $g_t=(\max_{1\leq k\leq t} \sup_{P_\phi} \sum_{n=k}^t z_\phi(s_n, a_n, s_{n+1}))^+$ (details regarding the implementation can be found in the appendix). The threshold $c$ in the detectors can be chosen according to the desired false alarm rate over a number of samples. For the CUSUM detector, for a probability of false alarm rate $\delta$ over $m$ samples, with $m> I^{-1} \ln \delta^{-1}$, we have that $c$ should satisfy $2me^{-c}=\delta$ to achieve the asymptotic lower bound as $\delta\to 0$ \cite{lai1998information}. For the attacks, we have chosen values of the constraints that yield a similar decrease in performance, that is, $\varepsilon=3$ in \cref{eq:problem_constrained_attack},  $\varepsilon=0.21$ in \cref{eq:simplified_stealthy_attack_problem}, and $\beta=6.2$ for the deterministic attack with penalty $\beta$.  The middle plot in \cref{fig:inventory_attack_evaluation} depicts the statistics $c_t,g_t$ for the attacks applied when the system had already converged to the stationary distribution ($t=0$ denotes the round $\nu$ at which the attack starts). We see that the orange curve, which corresponds to the attack in \cref{eq:simplified_stealthy_attack_problem}, is the least detectable one, and it takes roughly $3$ times more to detect this attack than the one in \cref{eq:problem_constrained_attack}, even though the performance decrease is similar. Finally, the right-most plot in \cref{fig:inventory_attack_evaluation} shows the goodness of approximating $\bar I$ with $\bar I_{\bar\gamma}$, where we computed $\phi^\star$ according to \cref{eq:simplified_stealthy_attack_problem} for different values of $(\varepsilon,\bar\gamma)$. As expected from \cref{proposition:discounted_information_term_error}, for large values of $\bar \gamma$ the two quantities coincide. Moreover, interestingly we observe that $\bar I_{\bar\gamma}\leq \bar I$ for every pair $(\varepsilon,\bar\gamma)$.
\begin{figure}[t]
	\centering
	\includegraphics[width=1\textwidth]{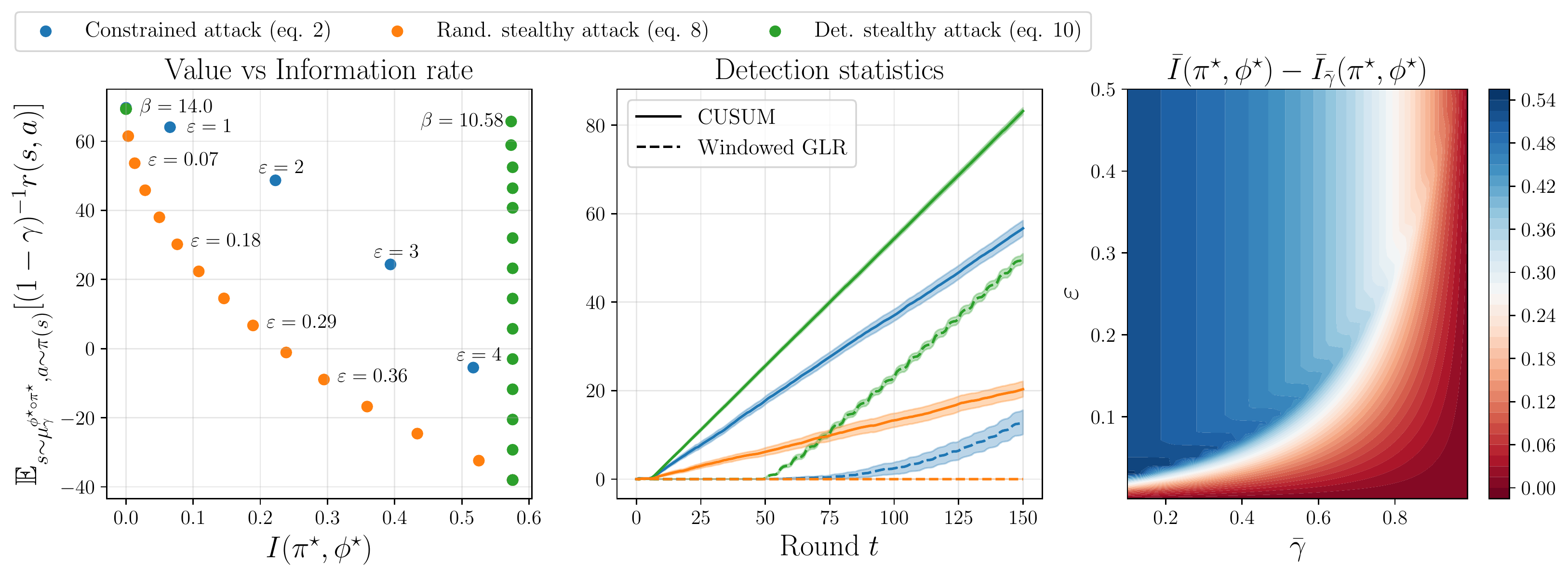}
	\caption{Attack evaluation of the inventory problem. From left to right: average normalized discounted reward of the victim vs the information rate $I$; detection statistics $c_t$ and $g_t$ for the optimal CUSUM detector and a windowed GLR rule; difference between  $\bar I$ and the discounted information rate $\bar I_{\bar \gamma}$ for $\phi^\star$, which was computed according to \cref{eq:simplified_stealthy_attack_problem} using different values of $(\varepsilon,\bar\gamma)$\protect\footnotemark[4].}
	\label{fig:inventory_attack_evaluation}
\end{figure}
\footnotetext[4]{Results were averaged over 100 simulations, and shadowed areas indicate  $99\%$ confidence interval.}
\subsection{Optimal attacks on linear dynamical systems}
\begin{figure}[b]
\centering

\subfloat[Statistics of the attack for different values of $\beta$]{\includegraphics[width=0.61\linewidth]{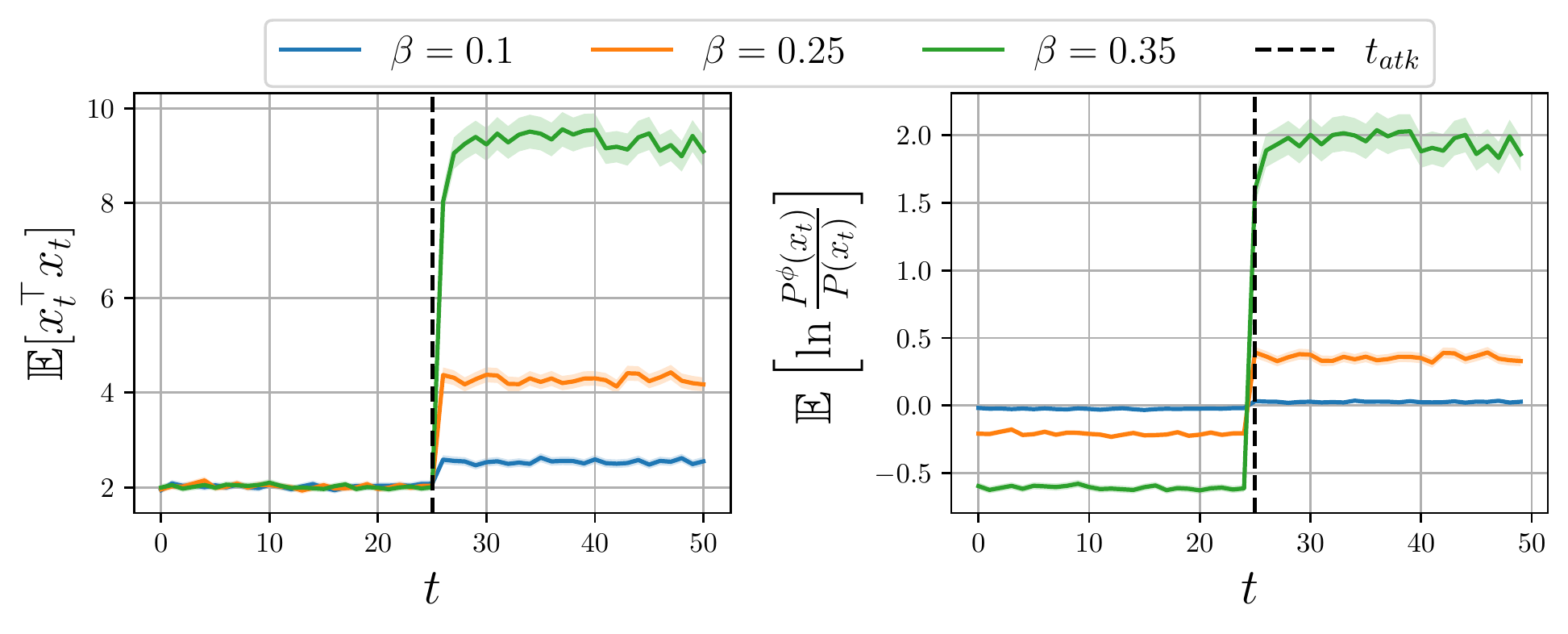}}
\quad
\subfloat[{Value of $I$ and $\mathbb{E}_{x\sim \mu}[x^\top x]$}]{\includegraphics[width=0.33\linewidth]{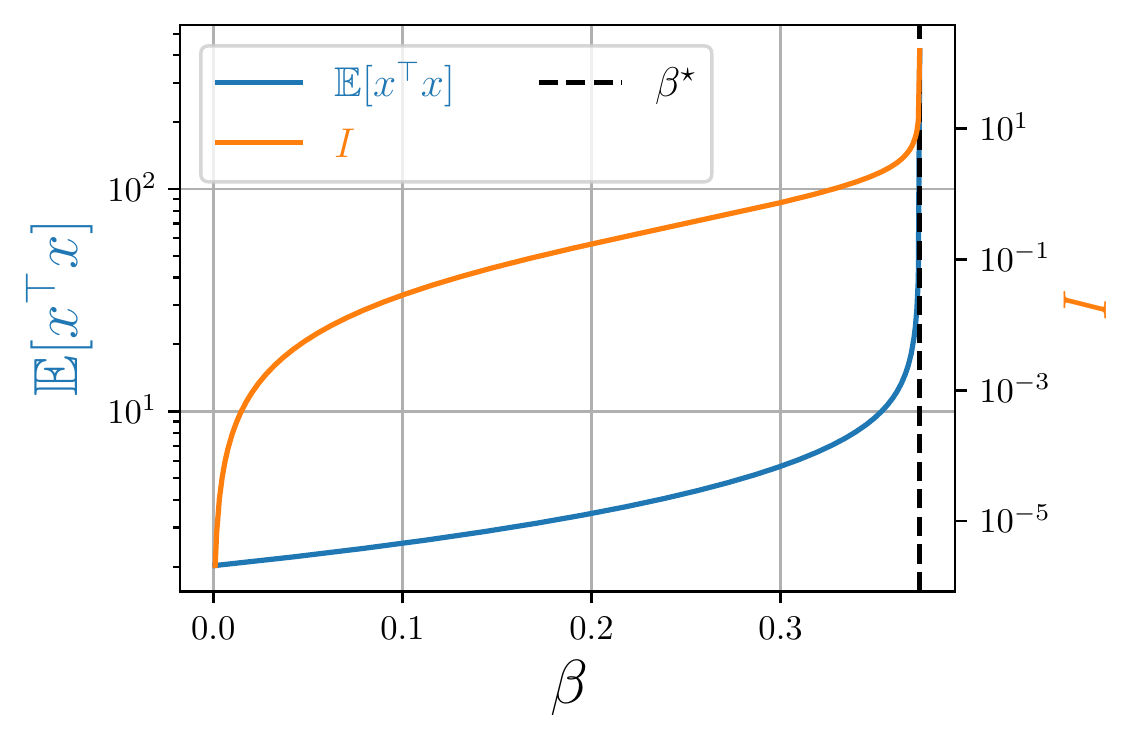}}

\caption{Attack simulation on a 2-dimensional system. The attack begins at $t_{atk}=25$ rounds.  On the left are shown the statistics of the attack for $\beta\in \{0.1, 0.25, 0.35\}$: the plots display the expected value of $x_t^\top x_t$ and  the log-likelihood ratio $z_t=\ln\frac{P^\phi(x_t)}{P(x_t)}$. On the right are depicted the asymptotic value of $\mathbb{E}_{x\sim \mu}[x^\top x]$ and $I$ as a function of $\beta$. Shadowed areas depict $95\%$ confidence interval.}
\label{fig:atk_linear_sys}
\end{figure}
We now turn our attention to linear dynamical systems.  Linear systems are of interests, since these are widely used models. Let us consider the following model
\[x_{t+1} = Ax_t + B(u_t+\bar u_t) + w_t,\quad x_0=0\]
where $x_n \in \mathbb{R}^n$ is the state at time $t$, $u_t\in \mathbb{R}^m$ is the control action, $\bar u_t \in \mathbb{R}^m$ is the attacker's action and $w_t$ is i.i.d. Gaussian noise, distributed according to $ \mathcal{N}(0, \Sigma), \Sigma\in\mathbb{S}^n_+$. For simplicity, assume $B$ is full-column rank, and assume the control policy of the victim is deterministic, of the type $u_t=Kx_t$, so that $L\coloneqq A+BK$ is Schur.

Assume that $\bar u_t$ is measurable with respect to the filtration $\sigma(x_0,\dots, x_t)$. We are interested in studying the adversarial problem with a penalty factor $\beta>0$ on the value  term:
\begin{equation}\label{eq:opt_problem_linear_sys}
 \min_{\bar u_0, \dots, \bar u_{T-1}} \frac{1}{T} \mathbb{E}\left[ \sum_{t=1}^{T-1}\ln\frac{P^{\bar u_t}(x_{t+1}|x_t, u_t)}{P(x_{t+1}|x_t, u_t)}-\beta \sum_{t=1}^T  x_t^\top x_t \right] \hbox{ s.t. } P^{\bar u}(x'|x,u)=\mathbb{E}_{\bar u}[P(x'|x,u+\bar u)]
\end{equation}
where $T$ is the horizon length, and $\beta>0$ is a penalty term that balances the trade-off between detectability and impact.

First, we note that the attack is not stealthy if the closed-loop system is unstable. This stability condition imposes some requirements on the set of possible values of $\beta$. Not surprisingly, we find that the constraint  depends on the noise level $\Sigma$. Secondly, as expected, we find that random attacks are, in general, better than deterministic attacks. Finally, we also note that directly optimizing over the distribution of $\bar u_t$ is a hard problem to solve. Indeed, we find that the first-order condition is an integral equation that does not admit a simple closed-form solution, unless one fixes a distribution family on $\bar u_t$. Therefore we study the problem of finding the optimal deterministic attack, and the optimal attack distributed according to a Gaussian distribution.

\begin{theorem}
	Let $J_d^\star$ be the solution of $\cref{eq:opt_problem_linear_sys}$ when $\bar u_t$ is a deterministic function of $x_t$, and $J_g^\star$ when $\bar u_t$ is distributed according to a Gaussian distribution $\mathcal{N}(\theta_t, V_t)$.
	Then, there exists $\beta^\star>0$ such that for all $\beta\in (0,\beta^\star)$ an optimal stealthy attack exists, and $J_g^\star,J_d^\star$ are both finite, satisfying $J_g^\star < J_d^\star$. 
	The optimal Gaussian attack is given by $\theta_t=\beta B^+ F_t^{-1}P_{t+1}L x_t$ and $V_t=\beta B^+ F_t^{-1}P_{t+1}\Sigma (B^+)^\top$,
	where $F_t= \left(\frac{1}{2}\Sigma^{-1}-\beta P_{t+1}\right) $, and $P_t$ satisfies
	\begin{equation}
		 P_t =I_n + L^\top P_{t+1} (I-2\beta \Sigma P_{t+1})^{-1}L.
	\end{equation}
	The value of $\beta^\star$ is given by $\beta^\star =\min(\beta_0,\beta_1)$, where $\beta_0=\inf\{\beta>0: \frac{1}{2}\Sigma^{-1}-\beta \bar P\prec  0\}$ and $\beta_1=\inf\{\beta>0: \frac{\beta}{2}\bar K^\top B^\top \Sigma^{-1}B\bar K - I\succ 0\}$, with $\bar P$ being the stationary solution of $P_t$ and $\bar K= B^+  \left(\frac{1}{2}\Sigma^{-1}-\beta \bar P\right)^{-1} \bar P L$.
\end{theorem}
Not surprisingly, from the theorem we have two immediate facts: (1) randomizing the attack benefits the adversary, which helps fooling the victim; (2) not all values of $\beta$ are feasible. As $\beta$ approaches $\beta^\star$, the attack becomes less stealthy and more impactful. Computing the optimal attack amounts to computing a Riccati-like recursive equation in $P_t$, which is a well-defined recursion only for $\beta<\beta^\star$.

\paragraph{Example.}  Here we analyse the impact of an optimal Gaussian attack on a $2$-dimensional linear system.  We consider a simple system with $A=\begin{bmatrix}0.7 & 0.9\\1.5 & 2
\end{bmatrix}, B=2\begin{bmatrix}
0 &1\\2 & 1
\end{bmatrix}, Q=I_2$. The feedback gain $K=-\begin{bmatrix}0.19 & 0.26125\\0.3325 &0.4275\end{bmatrix}$  guarantees that the closed-loop eigenvalues $\lambda_1,\lambda_2$ are approximately $\lambda_1\approx 0.001$ and $\lambda_2\approx 0.134$. In \cref{fig:atk_linear_sys} are shown the results of a Gaussian attack that starts after $t_{atk}=25$ rounds. We find that $\beta^\star \approx 0.373$, and as $\beta$ approaches $ \beta^\star$ the closed-loop eigenvalues converge to the boundary of the unit disk in the complex plane. This is also confirmed by the right plot of \cref{fig:atk_linear_sys}, which depicts what is the  value of $\mathbb{E}[\|x\|_2^2]$ at stationarity.  Moreover,we also have that $I$ increases as $\beta$ increases, making the attack less stealthy.

\section{Conclusions}\label{sec:conclusions}
In this work, we have introduced a new notion of stealthiness, based on information-theoretical quantities, that can be used to compute stealthy adversarial attacks on the control channel of a Markov Decision Process. The resulting maximization problem is, in general, hard to solve, due to the concavity of the arguments. Nonetheless, the problem can be solved by considering an upper bound on the detectability metric, which results in a problem whose optimal attack policy is stationary and randomized. Finally, we tested the proposed attack on the inventory control problem and a linear dynamical system. Numerical results for both cases confirmed the efficiency of the attack in decreasing performance as well as detectability. These results indicate the need for future work to study the problem of finding ways to make attacks more detectable. An interesting venue of research would be to study the max-min problem of two competing agents that, respectively, try to maximize, and minimize, performance and detectability. Additionally, another research direction is to extend the methods presented here to the case of attacks on the observations.

\bibliographystyle{apalike}
\bibliography{ref}

\begin{thebibliography}{}

\bibitem[Achiam et~al., 2017]{achiam2017constrained}
Achiam, J., Held, D., Tamar, A., and Abbeel, P. (2017).
\newblock Constrained policy optimization.
\newblock In {\em International Conference on Machine Learning}, pages 22--31.
  PMLR.

\bibitem[Altman, 1999]{altman1999constrained}
Altman, E. (1999).
\newblock {\em Constrained Markov decision processes}, volume~7.
\newblock CRC Press.

\bibitem[{\AA}str{\"o}m, 2012]{aastrom2012introduction}
{\AA}str{\"o}m, K.~J. (2012).
\newblock {\em Introduction to stochastic control theory}.
\newblock Courier Corporation.

\bibitem[Basseville et~al., 1993]{basseville1993detection}
Basseville, M., Nikiforov, I.~V., et~al. (1993).
\newblock {\em Detection of abrupt changes: theory and application}, volume
  104.
\newblock prentice Hall Englewood Cliffs.

\bibitem[Behzadan and Munir, 2017]{behzadan2017vulnerability}
Behzadan, V. and Munir, A. (2017).
\newblock Vulnerability of deep reinforcement learning to policy induction
  attacks.
\newblock In {\em Machine Learning and Data Mining in Pattern Recognition: 13th
  International Conference, MLDM 2017, New York, NY, USA, July 15-20, 2017,
  Proceedings}, volume 10358, page 262. Springer.

\bibitem[Berner et~al., 2019]{berner2019dota}
Berner, C., Brockman, G., Chan, B., Cheung, V., D{\k{e}}biak, P., Dennison, C.,
  Farhi, D., Fischer, Q., Hashme, S., Hesse, C., et~al. (2019).
\newblock Dota 2 with large scale deep reinforcement learning.
\newblock {\em arXiv preprint arXiv:1912.06680}.

\bibitem[Borkar, 2009]{borkar2009stochastic}
Borkar, V.~S. (2009).
\newblock {\em Stochastic approximation: a dynamical systems viewpoint},
  volume~48.
\newblock Springer.

\bibitem[Burke et~al., 2019]{burke2019gartner}
Burke, B., Cearley, D., Jones, N., Smith, D., Chandrasekaran, A., Lu, C., and
  Panetta, K. (2019).
\newblock Gartner top 10 strategic technology trends for 2020-smarter with
  gartner.

\bibitem[Carlini and Wagner, 2017]{carlini2017towards}
Carlini, N. and Wagner, D. (2017).
\newblock Towards evaluating the robustness of neural networks.
\newblock In {\em 2017 ieee symposium on security and privacy (sp)}, pages
  39--57. IEEE.

\bibitem[Chen et~al., 2019]{chen2019adversarial}
Chen, T., Liu, J., Xiang, Y., Niu, W., Tong, E., and Han, Z. (2019).
\newblock Adversarial attack and defense in reinforcement learning-from ai
  security view.
\newblock {\em Cybersecurity}, 2(1):1--22.

\bibitem[Chen et~al., 2018]{chen2018optimal}
Chen, Y., Norford, L.~K., Samuelson, H.~W., and Malkawi, A. (2018).
\newblock Optimal control of hvac and window systems for natural ventilation
  through reinforcement learning.
\newblock {\em Energy and Buildings}, 169:195--205.

\bibitem[Chow et~al., 2017]{chow2017risk}
Chow, Y., Ghavamzadeh, M., Janson, L., and Pavone, M. (2017).
\newblock Risk-constrained reinforcement learning with percentile risk
  criteria.
\newblock {\em The Journal of Machine Learning Research}, 18(1):6070--6120.

\bibitem[Deng et~al., 2016]{deng2016deep}
Deng, Y., Bao, F., Kong, Y., Ren, Z., and Dai, Q. (2016).
\newblock Deep direct reinforcement learning for financial signal
  representation and trading.
\newblock {\em IEEE transactions on neural networks and learning systems},
  28(3):653--664.

\bibitem[Diamond and Boyd, 2016]{diamond2016cvxpy}
Diamond, S. and Boyd, S. (2016).
\newblock {CVXPY}: {A} {P}ython-embedded modeling language for convex
  optimization.
\newblock {\em Journal of Machine Learning Research}, 17(83):1--5.

\bibitem[Goodfellow et~al., 2014]{goodfellow2014explaining}
Goodfellow, I.~J., Shlens, J., and Szegedy, C. (2014).
\newblock Explaining and harnessing adversarial examples.
\newblock {\em arXiv preprint arXiv:1412.6572}.

\bibitem[Haarnoja et~al., 2018]{haarnoja2018soft}
Haarnoja, T., Zhou, A., Abbeel, P., and Levine, S. (2018).
\newblock Soft actor-critic: Off-policy maximum entropy deep reinforcement
  learning with a stochastic actor.
\newblock In {\em International conference on machine learning}, pages
  1861--1870. PMLR.

\bibitem[Harris et~al., 2020]{harris2020array}
Harris, C.~R., Millman, K.~J., van~der Walt, S.~J., Gommers, R., Virtanen, P.,
  Cournapeau, D., Wieser, E., Taylor, J., Berg, S., Smith, N.~J., Kern, R.,
  Picus, M., Hoyer, S., van Kerkwijk, M.~H., Brett, M., Haldane, A., del
  R{\'{i}}o, J.~F., Wiebe, M., Peterson, P., G{\'{e}}rard-Marchant, P.,
  Sheppard, K., Reddy, T., Weckesser, W., Abbasi, H., Gohlke, C., and Oliphant,
  T.~E. (2020).
\newblock Array programming with {NumPy}.
\newblock {\em Nature}, 585(7825):357--362.

\bibitem[Huang et~al., 2017]{huang2017adversarial}
Huang, S., Papernot, N., Goodfellow, I., Duan, Y., and Abbeel, P. (2017).
\newblock Adversarial attacks on neural network policies.
\newblock {\em arXiv preprint arXiv:1702.02284}.

\bibitem[Hunter, 2007]{Hunter:2007}
Hunter, J.~D. (2007).
\newblock Matplotlib: A 2d graphics environment.
\newblock {\em Computing in Science \& Engineering}, 9(3):90--95.

\bibitem[Jasso-Fuentes and Hernandez-Lerma, 2009]{jasso2009blackwell}
Jasso-Fuentes, H. and Hernandez-Lerma, O. (2009).
\newblock Blackwell optimality for controlled diffusion processes.
\newblock {\em Journal of applied probability}, 46(2):372--391.

\bibitem[Kato, 2013]{kato2013perturbation}
Kato, T. (2013).
\newblock {\em Perturbation theory for linear operators}, volume 132.
\newblock Springer Science \& Business Media.

\bibitem[Kluyver et~al., 2016]{jupyter}
Kluyver, T., Ragan-Kelley, B., P{\'e}rez, F., Granger, B., Bussonnier, M.,
  Frederic, J., Kelley, K., Hamrick, J., Grout, J., Corlay, S., Ivanov, P.,
  Avila, D., Abdalla, S., Willing, C., and development team, J. (2016).
\newblock Jupyter notebooks - a publishing format for reproducible
  computational workflows.
\newblock In Loizides, F. and Scmidt, B., editors, {\em Positioning and Power
  in Academic Publishing: Players, Agents and Agendas}, pages 87--90,
  Netherlands. IOS Press.

\bibitem[Kumar et~al., 2020]{kumar2020adversarial}
Kumar, R. S.~S., Nystr{\"o}m, M., Lambert, J., Marshall, A., Goertzel, M.,
  Comissoneru, A., Swann, M., and Xia, S. (2020).
\newblock Adversarial machine learning--industry perspectives.
\newblock {\em arXiv:2002.05646}.

\bibitem[Kurakin et~al., 2016]{kurakin2016adversarial}
Kurakin, A., Goodfellow, I., Bengio, S., et~al. (2016).
\newblock Adversarial examples in the physical world.

\bibitem[Lai, 1998]{lai1998information}
Lai, T.~L. (1998).
\newblock Information bounds and quick detection of parameter changes in
  stochastic systems.
\newblock {\em IEEE Transactions on Information Theory}, 44(7):2917--2929.

\bibitem[LeCun et~al., 2015]{lecun2015deep}
LeCun, Y., Bengio, Y., and Hinton, G. (2015).
\newblock Deep learning.
\newblock {\em nature}, 521(7553):436--444.

\bibitem[Lee et~al., 2020]{lee2020spatiotemporally}
Lee, X.~Y., Ghadai, S., Tan, K.~L., Hegde, C., and Sarkar, S. (2020).
\newblock Spatiotemporally constrained action space attacks on deep
  reinforcement learning agents.
\newblock In {\em Proceedings of the AAAI Conference on Artificial
  Intelligence}.

\bibitem[Lin et~al., 2017]{lin2017tactics}
Lin, Y.-C., Hong, Z.-W., Liao, Y.-H., Shih, M.-L., Liu, M.-Y., and Sun, M.
  (2017).
\newblock Tactics of adversarial attack on deep reinforcement learning agents.
\newblock In {\em Proceedings of the 26th International Joint Conference on
  Artificial Intelligence}, pages 3756--3762.

\bibitem[Lorden et~al., 1971]{lorden1971procedures}
Lorden, G. et~al. (1971).
\newblock Procedures for reacting to a change in distribution.
\newblock {\em The Annals of Mathematical Statistics}, 42(6):1897--1908.

\bibitem[Madry et~al., 2017]{madry2017towards}
Madry, A., Makelov, A., Schmidt, L., Tsipras, D., and Vladu, A. (2017).
\newblock Towards deep learning models resistant to adversarial attacks.
\newblock {\em arXiv preprint arXiv:1706.06083}.

\bibitem[Miller and Veinott, 1969]{miller1969discrete}
Miller, B.~L. and Veinott, A.~F. (1969).
\newblock Discrete dynamic programming with a small interest rate.
\newblock {\em The Annals of Mathematical Statistics}, 40(2):366--370.

\bibitem[Mnih et~al., 2013]{mnih2013playing}
Mnih, V., Kavukcuoglu, K., Silver, D., Graves, A., Antonoglou, I., Wierstra,
  D., and Riedmiller, M. (2013).
\newblock Playing atari with deep reinforcement learning.
\newblock {\em arXiv preprint arXiv:1312.5602}.

\bibitem[Mnih et~al., 2015]{mnih2015human}
Mnih, V., Kavukcuoglu, K., Silver, D., Rusu, A.~A., Veness, J., Bellemare,
  M.~G., Graves, A., Riedmiller, M., Fidjeland, A.~K., Ostrovski, G., et~al.
  (2015).
\newblock Human-level control through deep reinforcement learning.
\newblock {\em nature}, 518(7540):529--533.

\bibitem[Morimoto and Doya, 2005]{morimoto2005robust}
Morimoto, J. and Doya, K. (2005).
\newblock Robust reinforcement learning.
\newblock {\em Neural computation}, 17(2):335--359.

\bibitem[Papernot et~al., 2016]{papernot2016distillation}
Papernot, N., McDaniel, P., Wu, X., Jha, S., and Swami, A. (2016).
\newblock Distillation as a defense to adversarial perturbations against deep
  neural networks.
\newblock In {\em 2016 IEEE symposium on security and privacy (SP)}, pages
  582--597. IEEE.

\bibitem[Pattanaik et~al., 2018]{pattanaik2018robust}
Pattanaik, A., Tang, Z., Liu, S., Bommannan, G., and Chowdhary, G. (2018).
\newblock Robust deep reinforcement learning with adversarial attacks.
\newblock In {\em Proceedings of the 17th International Conference on
  Autonomous Agents and MultiAgent Systems}, pages 2040--2042.

\bibitem[Pinto et~al., 2017]{pinto2017robust}
Pinto, L., Davidson, J., Sukthankar, R., and Gupta, A. (2017).
\newblock Robust adversarial reinforcement learning.
\newblock In {\em International Conference on Machine Learning}, pages
  2817--2826. PMLR.

\bibitem[Pollak, 1985]{pollak1985optimal}
Pollak, M. (1985).
\newblock Optimal detection of a change in distribution.
\newblock {\em The Annals of Statistics}, pages 206--227.

\bibitem[Puterman, 2014]{puterman2014markov}
Puterman, M.~L. (2014).
\newblock {\em Markov decision processes: discrete stochastic dynamic
  programming}.
\newblock John Wiley \& Sons.

\bibitem[Russo and Proutiere, 2021]{russo2021optimal}
Russo, A. and Proutiere, A. (2021).
\newblock Towards optimal attacks on reinforcement learning policies.
\newblock In {\em 2021 American Control Conference (ACC)}, pages 4561--4567.
  IEEE.

\bibitem[Schulman et~al., 2017]{schulman2017proximal}
Schulman, J., Wolski, F., Dhariwal, P., Radford, A., and Klimov, O. (2017).
\newblock Proximal policy optimization algorithms.
\newblock {\em arXiv preprint arXiv:1707.06347}.

\bibitem[Silver et~al., 2016]{silver2016mastering}
Silver, D., Huang, A., Maddison, C.~J., Guez, A., Sifre, L., Van Den~Driessche,
  G., Schrittwieser, J., Antonoglou, I., Panneershelvam, V., Lanctot, M.,
  et~al. (2016).
\newblock Mastering the game of go with deep neural networks and tree search.
\newblock {\em nature}, 529(7587):484--489.

\bibitem[Sutton and Barto, 2018]{sutton2018reinforcement}
Sutton, R.~S. and Barto, A.~G. (2018).
\newblock {\em Reinforcement learning: An introduction}.
\newblock MIT press.

\bibitem[Szegedy et~al., 2013]{szegedy2013intriguing}
Szegedy, C., Zaremba, W., Sutskever, I., Bruna, J., Erhan, D., Goodfellow, I.,
  and Fergus, R. (2013).
\newblock Intriguing properties of neural networks.
\newblock {\em arXiv preprint arXiv:1312.6199}.

\bibitem[Szepesv{\'a}ri, 2010]{szepesvari2010algorithms}
Szepesv{\'a}ri, C. (2010).
\newblock Algorithms for reinforcement learning.
\newblock {\em Synthesis lectures on artificial intelligence and machine
  learning}, 4(1):1--103.

\bibitem[Tan et~al., 2020]{tan2020robustifying}
Tan, K.~L., Esfandiari, Y., Lee, X.~Y., Sarkar, S., et~al. (2020).
\newblock Robustifying reinforcement learning agents via action space
  adversarial training.
\newblock In {\em 2020 American control conference (ACC)}, pages 3959--3964.
  IEEE.

\bibitem[Tartakovsky et~al., 2014]{tartakovsky2014sequential}
Tartakovsky, A., Nikiforov, I., and Basseville, M. (2014).
\newblock {\em Sequential analysis: Hypothesis testing and changepoint
  detection}.
\newblock CRC Press.

\bibitem[Tessler et~al., 2019]{tessler2019action}
Tessler, C., Efroni, Y., and Mannor, S. (2019).
\newblock Action robust reinforcement learning and applications in continuous
  control.
\newblock In {\em International Conference on Machine Learning}, pages
  6215--6224. PMLR.

\bibitem[Veeravalli and Banerjee, 2014]{veeravalli2014quickest}
Veeravalli, V.~V. and Banerjee, T. (2014).
\newblock Quickest change detection.
\newblock In {\em Academic Press Library in Signal Processing}, volume~3, pages
  209--255. Elsevier.

\bibitem[Virtanen et~al., 2020]{2020SciPy-NMeth}
Virtanen, P., Gommers, R., Oliphant, T.~E., Haberland, M., Reddy, T.,
  Cournapeau, D., Burovski, E., Peterson, P., Weckesser, W., Bright, J., {van
  der Walt}, S.~J., Brett, M., Wilson, J., Millman, K.~J., Mayorov, N., Nelson,
  A. R.~J., Jones, E., Kern, R., Larson, E., Carey, C.~J., Polat, {\.I}., Feng,
  Y., Moore, E.~W., {VanderPlas}, J., Laxalde, D., Perktold, J., Cimrman, R.,
  Henriksen, I., Quintero, E.~A., Harris, C.~R., Archibald, A.~M., Ribeiro,
  A.~H., Pedregosa, F., {van Mulbregt}, P., and {SciPy 1.0 Contributors}
  (2020).
\newblock {{SciPy} 1.0: Fundamental Algorithms for Scientific Computing in
  Python}.
\newblock {\em Nature Methods}, 17:261--272.

\bibitem[Yu et~al., 2019]{yu2019reinforcement}
Yu, C., Liu, J., and Nemati, S. (2019).
\newblock Reinforcement learning in healthcare: A survey.
\newblock {\em arXiv preprint arXiv:1908.08796}.

\bibitem[Yuan et~al., 2019]{yuan2019adversarial}
Yuan, X., He, P., Zhu, Q., and Li, X. (2019).
\newblock Adversarial examples: Attacks and defenses for deep learning.
\newblock {\em IEEE transactions on neural networks and learning systems},
  30(9):2805--2824.

\bibitem[Zhang et~al., 2020]{zhang2020robust}
Zhang, H., Chen, H., Xiao, C., Li, B., Boning, D., and Hsieh, C.-J. (2020).
\newblock Robust deep reinforcement learning against adversarial perturbations
  on observations.
\newblock {\em arXiv preprint arXiv:2003.08938}.

\end{thebibliography}
\section{Appendix}
\subsection*{Broader impact}\label{sec:broader_impact}
Reinforcement learning has rapidly gained interest over the last years and has attracted interest from both the research community and the industry.
However, reinforcement learning is still under heavy development, and more work needs to be done to understand how reinforcement learning can be deployed in real-world settings. Furthermore, the fact that in the future adversarial attacks will be one of the main ways to attack artificial intelligence systems makes the problem studied in this work even more relevant.
The impact of applying reinforcement learning must then be thoroughly investigated before its deployment for real-world applications. In this paper, we considered the problem of balancing the detectability and efficiency of an attack, and quantified what is the hardness of detecting an attack for Markov Decision Processes. However, this work may have a negative impact on society, since we study the problem of making attacks less detectable. Nonetheless, it is necessary to understand this topic to come up with better detectors, and better policies. To this aim, we provided a first analysis of the detectability problem and showed how attacks can impact finite state-action space processes as well as linear dynamical systems. We believe the work presented here paves the way for many interesting research directions that will allow us to have a better understanding of how to design better policies as well as attack detectors.

\subsection*{Limitations of the work}
This work studies the detectability problem of attacks from a theoretical point of view, and limitations are mainly due to the set of assumptions made throughout the papers. An assumption is that the reward signal depends on the poisoned action $\bar a_t$, and not the original action $a_t$. However, as pointed out in \cref{sec:optimal_attack}, this assumption can be relaxed in the hypothesis testing problem by considering an observation $(r_t, s_t, a_t)$ (instead of $(s_t,a_t)$) and by defining the corresponding conditional probabilities. Another assumption made in the text is that the underlying Markov decision process is ergodic and converges to a unique stationary distribution. However, note that this assumption is satisfied by many dynamical systems, such as stable linear systems. Moreover, the methods presented here can be extended also to the case where the process has several stationary probability measures. Finally, the proofs we provide consider finite state-action spaces and can be extended to consider general measurable spaces.

\subsection{Value bound on policies under attack}
In this section we provide a proof for the bound in \cref{proposition:bound_value}. 
\begin{proof}[Proof of \cref{proposition:bound_value}.]
	Let $s\in S$ and write $V^{\phi\circ\pi}(s)$ 
	\begin{align*}
		V^{\phi\circ\pi}(s) &= \mathbb{E}_{a\sim \pi(\cdot|s)}\left[\mathbb{E}_{\bar a \sim \phi(\cdot|s,a)} \left[ r(s,a) + \gamma \mathbb{E}_{s'\sim P(\cdot|s,\bar a)}[V^{\phi \circ \pi}(s')] \right]\right],\\
		&= \mathbb{E}_{a\sim \pi(\cdot|s)}\left[ r(s,a) + \gamma \mathbb{E}_{\bar a \sim \phi(\cdot|s,a)} \left[ \mathbb{E}_{s'\sim P(\cdot|s,\bar a)}[V^{\phi \circ \pi}(s')] \right]\right].
	\end{align*}
	
	If we now consider  $V^{\pi}(s)-V^{\phi\circ\pi}(s)$ it follows that we can write
	\begin{align*}
		V^{\pi}(s)-V^{\phi\circ\pi}(s) &= \gamma\mathbb{E}_{a\sim \pi(\cdot|s)}\left[  \mathbb{E}_{s'\sim P(\cdot|s, a)}[V^{\pi}(s')]- \mathbb{E}_{\bar a \sim \phi(\cdot|s,a)} \left[ \mathbb{E}_{s'\sim P(\cdot|s,\bar a)}[V^{\phi \circ \pi}(s')] \right]\right],\\
		&=  \gamma\mathbb{E}_{a\sim \pi(\cdot|s)}\left[ \mathbb{E}_{\bar a \sim \phi(\cdot|s,a)} \left[\mathbb{E}_{s'\sim P(\cdot|s, a)}[V^{\pi}(s')]- \mathbb{E}_{s'\sim P(\cdot|s,\bar a)}[V^{\phi \circ \pi}(s')] \right]\right].
	\end{align*}
	Now, let $V^{\pi}(s) = V^{\phi\circ\pi}(s)-\Delta(s)$, and take the absolute value of the left hand-side. We can then derive the following inequalities
	\begin{small}
	\begin{align*}
		|V^{\pi}(s)-V^{\phi\circ\pi}(s)| &\leq\gamma\mathbb{E}_{a\sim \pi(\cdot|s)}\left[ \mathbb{E}_{\bar a \sim \phi(\cdot|s,a)} \left[ \left|\mathbb{E}_{s'\sim P(\cdot|s, a)}[V^{\pi}(s')]- \mathbb{E}_{s'\sim P(\cdot|s,\bar a)}[V^{\phi \circ \pi}(s')] \right|\right]\right],\\
		&=\gamma\mathbb{E}_{a\sim \pi(\cdot|s)}\left[ \mathbb{E}_{\bar a \sim \phi(\cdot|s,a)} \left[ \left|\mathbb{E}_{s'\sim P(\cdot|s, a)}[V^{\pi}(s')]- \mathbb{E}_{s'\sim P(\cdot|s,\bar a)}[V^{ \pi}(s') - \Delta(s')] \right|\right]\right],\\
		&\leq\gamma\mathbb{E}_{a\sim \pi(\cdot|s)}\left[ \mathbb{E}_{\bar a \sim \phi(\cdot|s,a)} \left[ 2 \frac{R^\star}{1-\gamma} \|P(\cdot|s, a) - P(\cdot|s,\bar a)\|_{TV} +\mathbb{E}_{s'\sim P(\cdot|s,\bar a)}[\Delta(s')] \right]\right],\\
		&\leq \max_{a,\bar a}\frac{2 R^\star\gamma}{1-\gamma}\|P(\cdot|s, a) - P(\cdot|s,\bar a)\|_{TV} + \gamma \|V^{\pi}-V^{\phi\circ\pi}\|_\infty.
	\end{align*}\end{small}
	The result follows from the last inequality, by taking $\gamma \|V^{\pi}-V^{\phi\circ\pi}\|_\infty$ to the left hand side and maximizing over $s$.
\end{proof}

\subsection{Information bounds and discounted information rate}
In this section we first discuss the case where the victim uses also the reward signal to detect an attack. Next, we provide proofs for the various propositions presented \cref{sec:optimal_attack}.  Lastly, we discuss why the original problem is non-convex.
\subsubsection{Detectability of stealthy attacks using the reward signal}
It is possible to augment the QCD argument in \cref{sec:optimal_attack} by assuming that the victim  is allowed to observe $(r_t,s_t,a_t)$ instead of just $(s_t,a_t)$. For simplicity, we assume that the agent observes the reward $r_t$ upon selecting an action $a_t$. Moreover, assume that $r_t$ is  a random variable distributed according to $q(s_t,a_t)$, where $a_t$ it the action taken on the MDP $M$. Note that if the reward depends on the action taken by the victim, and not the one taken by the adversary, then it is the same setting that we studied in the main body of the paper. Consequently, we must have that the reward depends on the action executed on the MDP.

Then, we consider  a sequence of non-i.i.d. observations $\{(r_t, s_t, a_t)\}_{t\geq 0}$, and  assume the conditional density of $(r_t,s_t,a_t)$ given the previous measurement is $F(\cdot|r_{t-1},s_{t-1}, a_{t-1})$ for $t<\nu$, and $Q(\cdot|r_{t-1},s_{t-1}, a_{t-1})$ otherwise. Observe that we have
\begin{align}
F(r_t,s_t,a_t|r_{t-1},s_{t-1}, a_{t-1}) &= \pi(a_t|s_t) q(r_t|s_t,a_t)  P(s_t|s_{t-1}, a_{t-1}),\\
Q(r_t,s_t,a_t|r_{t-1},s_{t-1}, a_{t-1}) &=  q^\phi(r_t|s_t,a_t) \pi(a_t|s_t) P^\phi(s_t|s_{t-1}, a_{t-1}),
\end{align}
where $q^\phi(r|s,a) = \mathbb{E}_{\bar a\sim \phi(\cdot|s,a)}\left[q(r|s,\bar a)\right]$. The agent needs to decide in each round if she is under attack. Consequently, her decision takes the form of a stopping time $T$ with respect to the filtration $(\mathcal{F}_t)_t = (\sigma(s_0,a_0,r_0,\dots, s_t,a_t,r_t))_t$. The log-likelihood ratio takes the following form
 \begin{align}
 	z_\phi(s_{t-1},a_{t-1},r_{t-1},s_{t}, a_{t}, r_{t}) &=  \ln\frac{Q(r_t,s_t,a_t|r_{t-1},s_{t-1}, a_{t-1})}{F(r_t,s_t,a_t|r_{t-1},s_{t-1}, a_{t-1})},\nonumber\\
 	&=\ln\frac{q^\phi(r_t|s_t,a_t) P^\phi(s_t|s_{t-1}, a_{t-1})}{ q(r_t|s_t,a_t)  P(s_t|s_{t-1}, a_{t-1})},\nonumber\\
 	&=\ln\frac{q^\phi(r_t|s_t,a_t)}{ q(r_t|s_t,a_t)} + \ln\frac{ P^\phi(s_t|s_{t-1}, a_{t-1})}{  P(s_t|s_{t-1}, a_{t-1})} .
\end{align}
This last expression shows that in order to have a well-posed problem we require the rewards to be randomized, otherwise the log-likelihood ratio may not be well-defined. In simple words, attacks can be easily detected.

This new equation of the log-likelihood ratio changes the definition of stealthy attack changes as follows.
\begin{definition}[Information-theoretical stealthy attack with reward signal]\label{def:stealthy_attack_it_reward}
Suppose the reward signal is provided by the MDP. For $\varepsilon\geq 0$ we define an attack policy $\phi$ to be $\varepsilon$-stealthy if $ I(\pi, \phi)\leq \varepsilon$, where $ I(\pi,\phi)$ is the information rate number
\begin{align}
	I(\pi,\phi) &= \mathbb{E}_{s\sim \mu^{\phi \circ \pi}, a\sim \pi(\cdot|s)} \left[ \KL(P^\phi(s,a), P(s,a)) + \KL(q^\phi(s,a), q(s,a))\right],
\end{align}
 $\KL(\cdot,\cdot)$ is the KL-divergence, and $\mu^{\phi \circ \pi}$ is the on-policy distribution induced by $\phi$ and $\pi$.
\end{definition}
Due to the linearity of the arguments, all the reasonings in the main body of the paper can be straightforwardly extended to this case. First, it is possible to show that we can derive an upper bound on $I(\pi,\phi)$ similar to the one that we show in \cref{lemma:upper_bound_I} (see all the proofs in the next subsection):
\begin{equation}
I(\pi,\phi) \leq\mathbb{E}_{s\sim \mu^{\phi \circ \pi}, a\sim \pi(\cdot|s), \bar a \sim \phi(\cdot|s,a)} \left[ \KL(P(s,\bar a), P(s,a)) + \KL(q(s,\bar a), q(s,a))\right] \eqqcolon \bar I(\pi,\phi).
\end{equation}
Additionally, we can similarly prove that we can approximate the information rate $I$ with a discounted one. Define
\begin{equation}\resizebox{\hsize}{!}{%
$
\bar I_{\bar \gamma}(s,a) = \mathbb{E}^{\phi\circ\pi}\left[\sum_{t\geq 0}{\bar \gamma}^t (1-\bar \gamma)\left(\KL(P(s_t, \bar a_t), P(s_t, a_t))+\KL(q(s_t, \bar a_t), q(s_t, a_t))\right) \Big|s_{0}=s, a_0=a\right],
$}
\end{equation}
then, we have that $\lim_{\bar \gamma\to 1}\bar I_{\bar \gamma}(s,a) = I(\pi,\phi)$ for every $(s,a)$. We conclude by saying that an optimal stealthy attack, according to the new upper bound $\bar I(\pi,\phi)$, is computed by solving the following linear program
\begin{equation}
\resizebox{\hsize}{!}{%
$
\max_{\phi \in \Phi'(P,\pi)}  \bar V_{\bar \gamma}^{\phi\circ\pi}(s,a),\textrm{ s.t. } \mathbb{E}_{s\sim \mu_{\bar \gamma}^{\phi\circ\pi},a\sim\pi(\cdot|s),\bar a \sim \phi(\cdot|s,a)}[\KL(P(s,\bar a), P(s,a)) + \KL(q(s, \bar a), q(s, a))] \leq \varepsilon.$}\end{equation}
Also in this case the optimal policy is stationary, Markov and randomized (see next subsection to see how to compute it).

\subsubsection{Proofs of \cref{sec:optimal_attack}}
 We start by providing a proof for \cref{lemma:upper_bound_I}. Next, we discuss \cref{proposition:information_lb_discounted}, \cref{proposition:discounted_information_term_error}, \cref{prop:stealthy_randomized_attack} and \cref{prop:mnimum_achievable_infrate}.
\paragraph{Upper bound on the information rate.} To prove  \cref{lemma:upper_bound_I} we make use of the following lemma with $\rho=\phi$.
\begin{lemma}\label{lemma:upper_bound_I2}
	Assume that $P(s,\bar a) \ll P(s,a)$ for every $(s,a,\bar a) \in \mathcal{C}(\pi,\phi)$. Let  $\rho$ be a probability measure that dominates $\phi$. Then, the information value $I(\pi,\phi)$ can be upper bounded as follows
	\begin{equation}
	\resizebox{\hsize}{!}{%
	$
		I(\phi,\pi) \leq \mathbb{E}_{s\sim \mu^{\phi \circ \pi}, a\sim \pi(\cdot|s)} \left[ \KL(\phi(s,a), \rho(s,a))\right]  + \mathbb{E}_{s\sim \mu^{\phi \circ \pi}, a\sim \pi(\cdot|s),\bar a\sim\phi(\cdot|s,a)}\left[\KL(P(s,\bar a), P(s,a))] \right].
		$}
	\end{equation}
\end{lemma}
\begin{proof}
	For simplicity, we show the proof for finite state spaces, although it can be proven for general finite measurable spaces. Let  $\rho$ be a probability measure that dominates $\phi$, i.e.,  $\phi \ll \rho$. Define the following measure using $\rho$: $P^\rho(s'|s,a) =\sum_{\bar a} P(s'|s,  a)\rho(\bar a|s,a) =P(s'|s,a)$. Remember that $I$ is the average KL-number between $P^\phi$ and $P$. Then, we can write
	\[
	I(\pi,\phi) = \mathbb{E}_{s\sim \mu^{\phi \circ \pi}, a\sim \pi(\cdot|s)} \left[ \KL(P^\phi(s,a), P(s,a))\right]  = \mathbb{E}_{s\sim \mu^{\phi \circ \pi}, a\sim \pi(\cdot|s)} \left[ \KL(P^\phi(s,a), P^\rho(s,a))\right].
	\]
	It follows that
	\begin{align*}
		I(\pi,\phi) &= \mathbb{E}_{s\sim \mu^{\phi \circ \pi}, a\sim \pi(\cdot|s)} \left[\sum_{s'} P^\phi(s'|s,a) \ln \frac{P^\phi(s'|s,a)}{P^\rho(s'|s,a)}\right],\\
		&=\mathbb{E}_{s\sim \mu^{\phi \circ \pi}, a\sim \pi(\cdot|s)} \left[ \sum_{s'} \left(\sum_{\bar a} P(s'|s,\bar a)\phi(\bar a|s,a)\right) \ln \frac{\sum_{\bar a}  P(s'|s,\bar a)\phi(\bar a|s,a)}{\sum_{\bar a}  P(s'|s, a)\rho(\bar a|s,a)} \right].
	\end{align*}
	From the last expression we can apply the log-sum inequality to obtain
	\[I(\pi,\phi) \leq \mathbb{E}_{s\sim \mu^{\phi \circ \pi}, a\sim \pi(\cdot|s)} \left[ \sum_{s'} \sum_{\bar a} P(s'|s,\bar a)\phi(\bar a|s,a) \ln \frac{  P(s'|s,\bar a)\phi(\bar a|s,a)}{ P(s'|s, a)\rho(\bar a|s,a)} \right] ,\]
	and, as a consequence
	\begin{align*}
		I(\pi,\phi) &\leq \mathbb{E}_{s\sim \mu^{\phi \circ \pi}, a\sim \pi(\cdot|s)} \left[ \KL(\phi(s,a), \rho(s,a)) + \mathbb{E}_{\bar a\sim\phi(\cdot|s,a)}[\KL(P(s,\bar a), P(s,a))] \right].
	\end{align*}
\end{proof}
\paragraph{Discounted information rate.}
We now give the proof of \cref{proposition:information_lb_discounted} and \cref{proposition:discounted_information_term_error}. The idea of the proofs is to consider the information rate $I$ as the gain of a Markov reward process, and use the Laurent Decomposition due to Miller and Veinott \cite{miller1969discrete} to relate $I$ to a discounted version of the information rate.
\begin{proof}[Proof of \cref{proposition:information_lb_discounted}.]
The idea is to consider the Markov reward process (MRP) $\{(x_t, r_t)\}_{t=\nu}^\infty$ where $r_t=r(x_t) =(1-\gamma)\KL(P_1(x_t), P_0(x_t))$. Since the chain is irreducible we have two consequences: (1) it converges to a stationary measure, and (2) the gain of such chain is constant for all $x\in \mathcal{X}$. First, observe that the gain $g=\mathbb{E}_{x\sim\mu}[r(x)]$ is exactly equal to $(1-\gamma)I = (1-\gamma)\lim_{N\to\infty} \frac{1}{N}\sum_{t=\nu}^{\nu+N} \ln\frac{P_1(x_{t}|x_{t-1})}{P_0(x_t|x_{t-1})}$. Due to a result of Miller and Veinott \cite{miller1969discrete} (see also \cite{puterman2014markov}, corollary 8.2.4), if we denote by $V(x)$ the bias of the MRP, since the rewards are bounded, we have that
\[
I_\gamma(x) = (1-\gamma)^{-1}g + V(x) + e_\gamma(x)= I + V(x) + e_\gamma(x),
\]
where  $e_\gamma(x)$ satisfies $\lim_{\gamma\to1}e_\gamma(x)=0$ for all $x$ (see \cite{puterman2014markov} Theorem 8.2.3). The conclusion follows by  noting that $V(x)$  converges to $0$ for $\gamma\to 1$.
\end{proof}
\begin{proof}[Proof of \cref{proposition:discounted_information_term_error}.]
In the following we denote by $I_\gamma$ the $|\mathcal{X}|$-dimensional vector representation of $I_\gamma(x), x\in\mathcal{X}$. We also denote by $P_1$ the $|\mathcal{X}|\times |\mathcal{X}|$ transition matrix for $t\geq \nu$. In light of \cref{proposition:information_lb_discounted} and  theorem 8.2.3 in \cite{puterman2014markov} we have that
\[
I_\gamma(P_1, P_0) = \mathbf{1}I + \sum_{t=\nu}^\infty (-1)^{t+1-\nu}\left(\frac{1-\gamma}{\gamma}H_{P_1} \right)^{t+1-\nu}d
\]
where $\mathbf{1}$ is the unit vector, $d$ is a $|\mathcal{X}|$-column vector whose $j$-th entry is $ d_j = \KL(P_1(x_j), P_0(x_j))$, for some enumeration of the state space, and $H_{P_1}$ is the deviation matrix, which satisfies
\[
H_{P_1} = (I-P_1 +P_1^\star)^{-1}(I-P_1^\star),\quad P_1^\star = \lim_{N\to \infty} \frac{1}{N} \sum_{t=0}^{N-1} P_1^{t}.
\]
 For recurrent and irreducible $P_1$ we have that $P_1^\star = \boldsymbol{1}\otimes\mu^\top$. Furthermore, in aperiodic Markov chains we also have that $H_{P_1} = \sum_{t=\nu}^\infty (P_1^{t-\nu} -P_1^\star)$, from which follows that $\|H_{P_1}\|_1 \leq \sum_{t=\nu}^\infty \max_{x\in \mathcal{X}} \|P_1^{t-\nu}(x)- \mu\|_{TV} \leq \sum_{t=\nu}^\infty L\theta^{t-\nu} = L/(1-\theta)$. Consequently, the series converges if
\[
\gamma > \frac{1}{1+(1-\theta)/L},
\]
where the r.h.s. is clearly a positive number in $(0,1)$. Let $\alpha = (1-\gamma)/\gamma$: it follows that the series is upper bounded by
\begin{align*}
	\left\|\sum_{t=\nu}^\infty (-1)^{t+1-\nu}\left(\frac{1-\gamma}{\gamma}H_{P_1} \right)^{t+1-\nu}d \right\|_1 &\leq \sum_{t=\nu}^\infty \left\|\left(\frac{1-\gamma}{\gamma}H_{P_1} \right)^{t+1-\nu}\right\|_1 \|d\|_\infty ,\\
	&\leq \alpha \frac{LD^\star}{1-\theta} \sum_{k=0}^\infty \left( \frac{\alpha L}{1-\theta} \right)^k,\\
	&\leq \frac{\alpha L  D ^\star}{1-\theta -\alpha L},\\
\end{align*}
where the last term is equal to $\frac{(1-\gamma) L D^\star}{\gamma(1-\theta) - (1-\gamma)L}$.
\end{proof}
\paragraph{Optimal stealthy attacks.}  We now provide a  proof of \cref{prop:stealthy_randomized_attack}.
\begin{proof}[Proof of \cref{prop:stealthy_randomized_attack}.]
Let $\mu_{\bar \gamma}^{\phi\circ\pi}$ be the discounted state distribution induced by $\phi\circ\pi$ for an initial state distribution $p_\nu$ (for unichain models this initial distribution can be arbitrary as long as the elements sum up to 1 \cite{puterman2014markov}). Then, we begin by observing that 
\begin{align*}
\bar I_{\bar \gamma}(s,a) &= \frac{1}{1-\bar \gamma}\sum_{s,a,\bar a}\mu_{\bar \gamma}^{\phi\circ\pi}(s)\pi(a|s)\phi(\bar a|s,a)(1-\bar{\gamma}) \KL(P(s,\bar a), P(s,a)),\\
&= \mathbb{E}_{s\sim \mu_{\bar \gamma}^{\phi\circ\pi},a\sim\pi(\cdot|s),\bar a \sim \phi(\cdot|s,a)}[\KL(P(s,\bar a), P(s,a))].
\end{align*}
Consequently, the problem of maximizing $\bar V_{\bar \gamma}^{\phi\circ \pi}$ while keeping $\bar I_{\bar \gamma}(s,a)\leq \varepsilon$ can be cast as the following problem
\begin{equation}
\max_{\phi \in \Phi'(P,\pi)}  \bar V_{\bar \gamma}^{\phi\circ\pi}(s,a),\textrm{ s.t. } \mathbb{E}_{s\sim \mu_{\bar \gamma}^{\phi\circ\pi},a\sim\pi(\cdot|s),\bar a \sim \phi(\cdot|s,a)}[\KL(P(s,\bar a), P(s,a)] \leq \varepsilon.
\end{equation}
For fixed $\pi$,  let $\xi \in \Delta(S\times A\times A)$, and, spefically, let $\xi(s,a,\bar a) = \mu_{\bar \gamma}^{\phi\circ\pi}(s)\pi(a|s)\phi(\bar a|s,a)$. Then, $\xi$ represents the discounted state-action distribution induced by $\phi\circ\pi$ with discount factor $\bar \gamma$, where the state is $(s,a)$. Since we have the same discount factor also in the objective term  we can make use of the same distribution $\xi$ to equivalently rewrite the previous problem as
\begin{equation}
\resizebox{1\hsize}{!}{%
$
\begin{aligned}
\min_{\xi \in \Delta(S\times A\times A)}& \quad  \frac{1}{1-\bar \gamma} \sum_{s,a,\bar a} \xi(s,a,\bar a) \bar r(s,a,\bar a)\\
\textrm{s.t.} \quad & \sum_{\bar a} \xi(s,a,\bar a) = (1-\bar \gamma)\alpha(s) + \bar \gamma \sum_{s',a',\bar a'} \pi(a|s)P(s|s',\bar a ')\xi(s,a',\bar a'),\quad \forall (s,a)\in S\times A\\
  & \sum_{s, a, \bar a} \xi(s,a,\bar a) \KL(P(s,\bar a), P(s,a))\leq \varepsilon, \\
  & \xi(s,a,\bar a) =0,\quad \forall (s,a,\bar a) \notin  \{(s,a,\bar a): P(s,\bar a) \ll P(s,a) \wedge \pi(a|s)>0\}.
\end{aligned}
$}
\end{equation}
Thanks to theorem 8.9.6 in \cite{puterman2014markov} we know there exists a solution $\xi^\star$ to the problem, and the optimal policy $\phi^\star$ is stationary and randomized, satisfying
$\phi^\star(\bar a|s,a) = \xi^\star(s,a,\bar a) / \sum_{\bar a}\xi^\star(s,a,\bar a)$ for every $(s,a)$. In case $\pi$ is deterministic, the problem can be  simplified to
\begin{equation}
\begin{aligned}
\min_{\xi \in \Delta(S\times A)}& \quad  \frac{1}{1-\bar \gamma} \sum_{s,\bar a} \xi(s,\bar a) \bar r(s,\pi(s),\bar a)\\
\textrm{s.t.} \quad & \sum_{\bar a} \xi(s,\bar a) = (1-\bar \gamma)\alpha(s) + \bar \gamma \sum_{s',\bar a'} P(s|s',\bar a ')\xi(s,\bar a'),\quad \forall s\in S \\
  & \sum_{s, \bar a} \xi(s,\bar a) \KL(P(s,\bar a), P(s,\pi(s)))\leq \varepsilon, \\
  & \xi(s,\bar a) =0,\quad \forall (s,\bar a) \notin  \{(s,\bar a): P(s,\bar a) \ll P(s,\pi(s))\}.
\end{aligned}
\end{equation}
\end{proof}
\paragraph{Hardness of detecting an attack.} Finally, note that \cref{prop:mnimum_achievable_infrate} can be easily solved by using the following linear program
\begin{equation}
\begin{aligned}
\min_{\xi \in \Delta(S\times A\times A)}& \quad  \sum_{s, a, \bar a} \xi(s,a,\bar a) \KL(P(s,\bar a), P(s,a))\\
\textrm{s.t.} \quad & \sum_{\bar a} \xi(s,a,\bar a) =  \sum_{s',a',\bar a'} \pi(a|s)P(s|s',\bar a ')\mu(s,a',\bar a'),\quad \forall (s,a)\in S\times A\\\
  & \sum_{s,a}\left( \sum_{\bar a} \xi(s,a, \bar a)\right) r(s,a) \leq \rho\\
  & \xi(s,a,\bar a) =0,\quad \forall (s,a,\bar a) \notin  \{(s,a,\bar a): P(s,\bar a) \ll P(s,a) \wedge \pi(a|s)>0\}.
\end{aligned}
\end{equation}
However, in case one needs to consider the discounted reward, it is possible to consider the following problem
\begin{equation}
\min_{\phi \in \Phi'(P,\pi)} \bar I(\pi,\phi) \textrm{ s.t. }\mathbb{E}_{s\sim \mu_{\gamma}^{\phi\circ\pi},a\sim\pi(\cdot|s)}[r(s,a)] \leq \rho
\end{equation}
where we considered the  discounted reward instead of the ergodic one through the discounted stationary distribution. Note, moreover, that the information rate is computed using the on-policy distribution $\mu^{\phi\circ\pi}$. To solve the problem one can rewrite it by considering the state-action distributions $\xi_0(s,a,\bar a) = \mu^{\phi\circ\pi}(s)\pi(a|s)\phi(\bar a|s,a)$ and $\xi_1(s,a,\bar a) =\mu_{\gamma}^{\phi\circ\pi}(s)\pi(a|s)\phi(\bar a|s,a)$. However, that results in a problem with non-convex constraints since the policy $\phi$ in each state $(s,a)$ must be the same, i.e., we require $\xi_0(s,a,\bar a) \|\xi_1(s,a)\|_1 = \xi_1(s,a,\bar a)\|\xi_0(s,a)\|_1$. A simple workaround is to approximate $\bar I$ using $\bar I_\gamma$, as long as $\gamma$ is sufficiently close to $1$. This yields the following problem
\begin{equation}
\min_{\phi \in \Phi'(P,\pi)} \bar I_\gamma(\pi,\phi) \textrm{ s.t. }\mathbb{E}_{s\sim \mu_{\gamma}^{\phi\circ\pi},a\sim\pi(\cdot|s)}[r(s,a)] \leq \rho.
\end{equation}
which can be computed by solving the following linear program
\begin{equation}
\begin{aligned}
\min_{\xi \in \Delta(S\times A\times A)}& \quad  \sum_{s, a, \bar a} \xi(s,a,\bar a) \KL(P(s,\bar a), P(s,a))\\
\textrm{s.t.} \quad &\sum_{\bar a} \xi(s,a,\bar a) = (1- \gamma)\alpha(s) +  \gamma \sum_{s',a',\bar a'} \pi(a|s)P(s|s',\bar a ')\xi(s,a',\bar a'),\quad \forall (s,a)\in S\times A\\
  & \sum_{s,a}\left( \sum_{\bar a} \xi(s,a, \bar a)\right) r(s,a) \leq \rho,\\
  & \xi(s,a,\bar a) =0,\quad \forall (s,a,\bar a) \notin  \{(s,a,\bar a): P(s,\bar a) \ll P(s,a) \wedge \pi(a|s)>0\}.
\end{aligned}
\end{equation}
\newpage
\section{Examples and numerical results}

\textbf{Hardware and software setup.} All experiments were executed on a stationary desktop computer,
featuring an Intel Xeon Silver 4110 CPU, 48GB of RAM and a GeForce GTX 1080 graphical card. Ubuntu 18.04 was installed on the computer.

\textbf{Code and libraries.} The code is released with the MIT license. Please, check the README file for instructions to run the code. Python 3.5 is required to run the code, as well as the following libraries: NumPy \cite{harris2020array}, SciPy \cite{2020SciPy-NMeth}, Matplotlib \cite{Hunter:2007}, CVXPY \cite{diamond2016cvxpy} and Jupyter Notebook \cite{jupyter}. Simulations take approximately 1 day to run.

\subsection{The inventory control problem}
\paragraph{Description of the example.} The inventory control problem is a widely known problem in literature (see, e.g., \cite{szepesvari2010algorithms}), and concerns the problem of managing an inventory of fixed maximum size $N$ in  face of uncertain demand. In each round the agent must decide the amount of items to be ordereded for the next day. The cost of purchasing $a_t$ items is $k\indi_{\{{a_t>0}\}}+ca_t$, where $k>0$ is a fixed cost of ordering nonzero items, and $c>0$ is a fixed unitary price. Upon selling $\ell$ items the agent is paid an amount of $ p\ell$, where $p>0$ is the price of a single item. Finally, there is also a cost of holding an inventory of size $s>0$, that is $hs$, with $h>0$ and $p>h$. The demand $d_t$ at time $t$ is modeled according to a Poisson distribution, with demand rate $\lambda$. Then, given $s_t$ and $a_t$, the size of the inventory the next round it $s_{t+1} = \max(0, \min(N, s_t+a_t) - d_{t+1} )$,  with reward $r(s_t, a_t, s_{t+1}) = -k\indi_{\{{a_t>0}\}}-hx_t -c\max(0, \min(N, s_t+a_t) - x_{t} ) + p\max(0, \min(N, s_t+a_t) - x_{t+1} ) $. To run the simulations, we have chosen $N=35$, $k=3, c=2, h=2, p=4, \lambda=6$. We used $\gamma=\bar \gamma=0.95$ to compute both the agent's policy and the adversary's policy. The attacks were applied after the system had already converged to the stationary distribution, after $\nu=25$ steps. Results were averaged over 100 simulations, and shadowed area indicate a confidence interval of $99\%$ probability.

\paragraph{Attack detection.} We also evaluated the detectability of these attacks using the optimal CUSUM detector $T_c=\inf\{t: \max_{1\leq k\leq t} \sum_{n=k}^t z_\phi(s_n, a_n, s_{n+1}) \geq c\}$, and a Generalized Likelihood Ratio (GLR) rule $T_g =\inf\{t:  \max_{1\leq k\leq t} \sup_{P_\phi} \sum_{n=k}^t z_\phi(s_n, a_n, s_{n+1}) \geq c\}$. To implement the GLR rule we estimate the transition kernel $P^\phi$, and used a window-limited GLR rule \cite{lai1998information} with $38$ parallel statistics, with a delay of $5$ samples between each statistics. Specifically, the $n$-th statistic computes an estimate according to the last $5n$ samples. 

\subsection{Optimal attack on linear dynamical systems}

We are interested in the following systems
\[x_{t+1} = Ax_t + Ba_t + w_t\]
where $x_0=0$, $B\in \mathbb{R}^{n\times m}$ is full column-rank and $w_t\sim \mathcal{N}(0,\Sigma)$.  We assume for simplicity that the adversarial policy $\phi$ is additive in the control action, so that $a_t = u_t+\bar u_t$, where $u_t$ is the main agent's control action and $\bar u_t$ is the adversarial's action. We assume $u_t = Kx_t$, where $K$ is computed according to standard control techniques (e.g., LQR), and that $\bar u_t$ is a random variable measurable with respect to the sigma algebra $\sigma(x_t)$ , and we can write that $\bar u_t \sim \phi(x_t)$ (since  $u_t$ is deterministic, it suffices to consider random variables measurable with respect to $\sigma(x_t)$).

We are interested in the following finite-horizon optimization problem

\[
\min_{\bar u_0, \dots, \bar u_{T-1}} \frac{1}{T} \mathbb{E}\left[\sum_{t=1}^{T-1}\ln\frac{P^\phi(x_{t+1}|x_t, u_t)}{P(x_{t+1}|x_t, u_t)}- \sum_{t=1}^T \beta x_t^\top x_t \right].
\]
We can rewrite the previous objective by noting that the first quantity is an expectation of KL-divergences
\begin{align*}\mathbb{E}\left[\sum_{t=1}^{T-1}\ln\frac{P^{\bar u_t}(x_{t+1}|x_t, u_t)}{P(x_{t+1}|x_t, u_t)}\right]&=\sum_{t=1}^{T-1}\mathbb{E}\left[\mathbb{E}_{x}\left[\ln\frac{P^{\bar u_t}(x|x_t, u_t)}{P(x|x_t, u_t)}\Big|x_t,u_t\right]\right],\\
&=\sum_{t=1}^{T-1}\mathbb{E}\left[D(P^{\bar u_t}(x_t,u_t), P(x_t,u_t))\right]
\end{align*}
Since the control action $u_t$ is a deterministic function of $x_t$, we write $D(P^{\bar u_t}(x_t), P(x_t))$ in the following.
Letting $I_t = D(P^{\bar u_t}(x_t), P(x_t))$ we can write
\[
\min_{\bar u_0, \dots, \bar u_{T-1}} \frac{1}{T} \mathbb{E}\left[\sum_{t=1}^{T-1} I_t- \beta\sum_{t=1}^T  x_t^\top x_t \right].
\]
Solving the optimization problem is not straightfoward, due to the dependency of $P^{\bar u}$ on the random variable $\bar u$. To show the hardness of solving such problem, we shall take a dynamic programming approach. 

We also state the following simple lemma that will be useful in the calculations.
\begin{lemma}[Lemma 3.3 in \cite{aastrom2012introduction}]\label{lemma:expectation_mean_square}
	Let $x \in \mathbb{R}^n$ be a normal random variable with mean $\mu$ and covariance $\Sigma$. Then, for any $n$-square matrix $S$ we have
	\begin{equation}
		\mathbb{E}[x^\top S x] = \mu^\top S \mu + \trace(S\Sigma).
	\end{equation}
\end{lemma}

\subsection{Deterministic optimal attacks}
We first consider the case where $\bar u_t$ is a deterministic function of $x_t$. Consider a dynamic programming approach, and define
 \[
 J_t^\star(x_{t})=\min_{\bar u_{t}} \mathbb{E}[I_{t}- \beta x_{t}^\top x_{t} + J_{t+1}^\star(x_{t+1})|x_{t}] 
 \]
 with $J_T^\star(x_T) = -\beta x_T^\top x_T$. Moreover, note that $I_t$ has the following expression for every $t$: $I_t = \KL(P(x_t, \bar u_t), P(x_t, u_t)) = \frac{1}{2} \bar u_t^\top B^\top \Sigma^{-1}B \bar u_t$. 
 
We now prove by induction that $J_t^\star(x) = -\beta (x^\top P_t x + p_t)$, with $P_t \succ 0$ and $p_t\geq 0$. At time $t=T$ we simply have  $P_T=I_n$ and $p_T=0$. 

\paragraph{Step $t=T-1$.} At time $T-1$ we have
 \begin{align*}
 J_{T-1}^\star(x_{T-1})&=\min_{\bar u_{T-1}} \mathbb{E}\left[\frac{1}{2} \bar u_{T-1}^\top B^\top \Sigma^{-1}B \bar u_{T-1} -\beta x_{T-1}^\top x_{T-1}+ J_{T}^\star(x_{T}) |x_{T-1}\right],\\
 &=\min_{\bar u_{T-1}} \mathbb{E}\Big[\frac{1}{2} \bar u_{T-1}^\top B^\top \Sigma^{-1}B \bar u_{T-1} -\beta (x_{T-1}^\top (I+L^\top L)x_{T-1} +\trace\Sigma\\
 &\qquad +\bar u_{T-1}^\top B^\top B \bar u_{T-1} +2x_{T-1}^\top L^\top B \bar u_{T-1}) |x_{T-1}\Big],\\
 &=\min_{\bar u_{T-1}} \mathbb{E}\Big[
 -\beta x_{T-1}^\top ( I_n + L^\top L )x_{T-1} +\bar u_{T-1}^\top B^\top \left(\frac{1}{2}   \Sigma^{-1}-\beta I_n \right)B u_{T-1}
 \\
  &\qquad -2\beta x_{T-1}^\top L^\top B \bar u_{T-1}  -\beta \trace \Sigma|x_{T-1}\Big].
 \end{align*}
 Note that in the second equality we made use of the fact that $\mathbb{E}[J_T^\star(x_T)|x_{T-1}] = -\beta \mathbb{E}[x_T^\top x_T|x_{T-1}] =  -\beta \mathbb{E}[(Lx_{T-1}+B\bar u_{T-1}+w_{T-1})^\top (Lx_{T-1}+B\bar u_{T-1}+w_{T-1})|x_{T-1}]$, and then used \cref{lemma:expectation_mean_square}.

 Now, let $y_t=B\bar u_t$ and $F_{T-1} = \frac{1}{2}   \Sigma^{-1}-\beta I_n $, then
 \begin{align*}
  J_{T-1}^\star(x_{T-1})&=\min_{\bar u_{T-1}} \mathbb{E}\Big[
  -\beta x_{T-1}^\top ( I_n + L^\top L )x_{T-1} + y_{T-1}^\top F_{T-1}(y_{T-1} - \beta  F_{T-1}^{-1}Lx_{T-1}) \\
  &\qquad -\beta y_{T-1}^\top L x_{T-1} -\beta\trace \Sigma|x_{T-1}\Big].
 \end{align*}
 Add $\pm \beta^2 x^\top L^\top F^{-1} L x_{T-1}$ to get
  \begin{align*}
   &=\min_{\bar u_{T-1}} \mathbb{E}[
   -\beta x_{T-1}^\top ( I_n + L^\top L + \beta L^\top F_{T-1}^{-1} L)x_{T-1} \\
   &\qquad+ (y_{T-1}-\beta F_{T-1}^{-1} L x_{T-1})^\top F_{T-1} ( y_{T-1} - \beta F_{T-1}^{-1}  Lx_{T-1})-\beta \trace \Sigma|x_{T-1}].
  \end{align*}
 If we impose $0<\beta<\lambda_{min}[(2\Sigma)^{-1}]= (2\lambda_{max}(\Sigma))^{-1}$, then there is a unique  minimum, and $F_{T-1}$ is invertible. The solution is given by $\bar u_{T-1}=  B^+ y_{T-1}$, where $
 y_{T-1} =  \beta F_{T-1}^{-1} L x_{T-1}$. Thus
 \[
 J_{T-1}^\star(x_{T-1}) = -\beta x_{T-1}^\top ( I_n + L^\top L + \beta L^\top F_{T-1}^{-1} L)x_{T-1} -\beta \trace \Sigma
 \]
 The formula for $J_t^\star$ clearly holds for $T-1$ with $P_{T-1} =  I_n +  L^\top L +\beta L^\top F_{T-1}^{-1} L$ and $p_{T-1} =  \trace\Sigma$.
 
 \paragraph{Induction step.}  
   Then, proceeding by induction, assuming that the formula for $J_t^\star$ holds for $t+1$ we show that it holds also for $t$.
 First observe
 \begin{align*}
 \mathbb{E}[J_{t+1}^\star(x_{t+1})|x_t] &= \mathbb{E}[-\beta(x_{t+1}^\top P_{t+1} x_{t+1} + p_{t+1})|x_t],\\
 &=-\beta[(Lx_t+B\bar u_t)^\top P_{t+1}(Lx_t+B\bar u_t) + \trace(\Sigma P_{t+1}) + p_{t+1}],
 \end{align*}
 then
 \begin{align*}
  J_{t}^\star(x_{t})&=\min_{\bar u_{t}} \mathbb{E}\left[\frac{1}{2} \bar u_{t}^\top B^\top \Sigma^{-1}B \bar u_{t} -\beta x_{t}^\top x_{t}+ J_{t}^\star(x_{t+1}) |x_{t}\right],\\
  &=\min_{\bar u_{t}} \mathbb{E}\Big[\frac{1}{2} \bar u_{t}^\top B^\top \Sigma^{-1}B \bar u_{t} -\beta (x_{t}^\top (I+L^\top P_{t+1}L)x_{t} +\trace(\Sigma P_{t+1})+p_{t+1}\\
  &\qquad +\bar u_{t}^\top B^\top P_{t+1}B \bar u_{t} +2x_{t}^\top L^\top  P_{t+1}B \bar u_{t}) |x_{t}\Big],\\
  &=\min_{\bar u_{t}} \mathbb{E}\Big[
  -\beta x_{t}^\top ( I_n + L^\top  P_{t+1}L )x_{t} +\bar u_{t}^\top B^\top \left(\frac{1}{2}   \Sigma^{-1}-\beta P_{t+1} \right)B u_{T-1}
  \\
   &\qquad -2\beta x_{T-1}^\top L^\top P_{t+1}B \bar u_{T-1}  -\beta (p_{t+1}+\trace (\Sigma P_{t+1}))|x_{T-1}\Big].
  \end{align*}
Define $F_t =  \left(\frac{1}{2}\Sigma^{-1} - \beta P_{t+1}\right)$. Similarly to before, by introducing $y_t=B\bar u_t$ we obtain
  \begin{align*}
   J_{t}^\star(x_{t})&=\min_{\bar u_{t}} \mathbb{E}\Big[
   -\beta x_{t}^\top ( I_n + L^\top P_{t+1}L +\beta L^\top P_{t+1} F_t^{-1} P_{t+1} L)x_{t} \\
   &\qquad+ (y_{t}-\beta F_t^{-1} P_{t+1}L x_{t})^\top F_t ( y_{t} - \beta F_t^{-1}  P_{t+1}Lx_{t})-\beta (p_{t+1}+\trace (\Sigma P_{t+1}))|x_{t}\Big].
  \end{align*}
  The solution exists and is unique if $F_t\succ 0$, that is, we need $ \beta\leq \frac{1}{2 \|\Sigma\|_2 \|P_{t+1}\|_2}$. 
  
  Then, the solution at time $t$ is
  \[
  \bar u_t = \beta B^+F_t^{-1}  P_{t+1}Lx_{t}
  \]
  and the cost becomes  
  \[
  J_{t}^\star(x_{t}) = -\beta x_{t}^\top ( I_n + L^\top P_{t+1}L +\beta L^\top P_{t+1} F_t^{-1} P_{t+1} L)x_{t} - \beta (p_{t+1}+\trace (\Sigma P_{t+1}))
  \]
  where
  \begin{align*}
	  P_{t} &=  I_n + L^\top P_{t+1}L +\beta L^\top P_{t+1}F_t^{-1} P_{t+1}L,\\
	  p_t &=  p_{t+1}+\trace(\Sigma S_{t+1})
  \end{align*}
  This proves the induction. Moreover, we observethat $P_t$ is positive definite if $P_{t+1}$ is positive definite and $\beta$ satisfies the condition $\beta\leq \frac{1}{2 \|\Sigma\|_2 \|P_{t+1}\|_2}$ for every $t$.
  
   \paragraph{Compact expression for $P_t$.}  
  Note that the equation of $P_t$ can be written in a compact. First, write 
    $
  	  P_t = I_n + L^\top P_{t+1}(I_n +\beta F_{t+1}^{-1} P_{t+1})L
    $
    and use the identity
    \[
  	  (U^{-1}+VZ^{-1}W)^{-1}=U-UV(Z+WUV)^{-1}WU.
    \]
    By setting $U=I, V=\beta , W= P_{t+1},Z=F_{t+1}$ we obtain
    \[
    P_t = I_n + L^\top P_{t+1}\left(I-2\beta\Sigma P_{t+1}\right)^{-1}L.
    \]

  \paragraph{Value of $\beta$.} The dependency on $\beta$ of  $P_t$ does not make it clear how $\beta$ should be chosen, especially when the horizon $T$ goes to infinity. To help the analysis, we write $P_{t,T}$ to also highlight the dependency on the horizon.

  To conduct the analysis we do the following:
  \begin{enumerate}
	  \item[(1)] Observe that for small values of $\beta$ there exists ${\cal B}$ such that $\beta \in {\cal B}$ makes the recursion of $P_{t,T}$  well defined.
	  \item[(2)] If the recursion if well defined, then $ P_{0,T} \succeq P_{t,T} \succ 0$ for every $t$, and $P_{t,T} \preceq P_{t,T+1}$.
	  \item[(3)] Then, if $\beta$ satisfies $\beta\leq \frac{1}{2 \|\Sigma\|_2 \|P_{0,T+1}\|_2}$ then it also satisfies $\beta\leq \frac{1}{2 \|\Sigma\|_2 \|P_{t,T}\|_2}$.
	  \item[(4)] Let the horizon $T\to\infty$, and find the value of $\beta^\star$ for which the stationary solution $\bar P$ is well-defined.
	  \item[(5)] We conclude by observing that $\bar P \succeq P_{t,T}$ for every $t$ and $T$. Therefore, because of (3) any value of $\beta \in (0,\beta^\star)\subset {\cal B}$ makes the recursion well-defined for any horizon $T$.
  \end{enumerate}

  (1) First, note that for $\beta \to 0$ then $P_{t,T}$ converges to the solution of the Lyapunov equation $P_{t,T} = I_n +L^\top P_{t+1,T} L$  for any $(t,T)$. By continuity, there exists a neighborhood ${\cal B}$ of $\beta$ for which  $P_{t,T}$ exists for every $\beta \in {\cal B}$ (to show this  we can also employ Theorem 6.8 in \cite{kato2013perturbation}, which states that for a symmetric operator $T(x)$, continuous and differentiable, also the eigenvalues are $C^1$ function of $x$.)
  
  (2) For the recursion to be well-defined we therefore need $I_t-\beta x_t^\top x_t$ to be negative definite (otherwise the induction fails, and for $\beta\to 0$ the solution $P$ converges to a negative definite matrix, which is not possible since it converges to the Lyapunov solution of the unperturbed system).
    Then, if this condition is satisfied,  for every $\beta \in {\cal B}$ we have $P_{t+1,T} \preceq P_{t,T}$. This follows from the simple fact that if  $I_{t+1} -\beta x_{t+1}^\top x_{t+1}$ is negative definite then $J_{t}^\star(x) \leq  J_{t+1}^\star(x)$, which implies $P_{t+1,T} \preceq P_{t,T}$ .   Next, observe that $P_{t,T+1} \succeq P_{t,T}$ (follows easily by analyzing the recursion in the previous section).
  
  (4-5) Therefore, we let $T\to\infty$ and study the stationary solution to understand what is the maximum value of $\beta$.
 Define the steady state Riccati equation
  \begin{align*}
      \bar P&=(I_n+L^\top \bar PL + \beta L^\top  \bar P  F^{-1}  \bar PL),\\
      F&= \frac{1}{2}\Sigma^{-1}-\beta \bar P
  \end{align*}
  At this point, to find the maximum value of $\beta$ we need to find the minimum value of $\beta$ for which the recursion is not well defined. Let $\bar K=B^+F^{-1} \bar P L$, and define $\beta_0=\inf\{\beta>0: \frac{1}{2}\Sigma^{-1}-\beta \bar P\prec  0\}$ and $\beta_1=\inf\{\beta>0: \frac{\beta}{2}\bar K^\top B^\top \Sigma^{-1}B\bar K - I\succ 0\}$.  From which follows that $\beta^\star$ is given by $\beta^\star =\min(\beta_0,\beta_1)$.

\subsection{Gaussian optimal attacks}
The previous discussion on $\beta$ follows from the fact that the adversary just prefers to make the system unstable for large values of $\beta$. For large values of $\beta$ it is simply impossible not to be detected, therefore the adversary prefers to make the system unstable.  We wonder if this can be changed by considering a random attack. 

Moreover, we also wonder if random attacks are in general better than deterministic attacks.

Minimizing $J^\star_t(x_t)$ over some distribution $\phi_t(x_t)$ from which $\bar u_t$ is drawn from can't be easily solved, since it involves solving an integral equation.
However, we can impose a parametrized distribution on $\phi_t$ and solve for the parameters. We can for example impose that $\bar u_t \sim \mathcal{N}(\theta_t, V_t)$.

In this case $I_t$ is computed as follows
\[
I_t = \frac{1}{2} \left[\trace(\Sigma^{-1}BV_tB^\top)+\theta_t^\top B^\top \Sigma^{-1} B\theta_t - \ln|I+\Sigma^{-1}BV_tB^\top|\right]
\]
Assume again that $J_t^\star(x) = -\beta x^\top P_t x -p_t$, where $P_t\succ 0$ and $p_t\geq 0$. It clearly holds at time $T$ for $P_T=\beta I, p_T=0$. Finally, for simplicity, let $R_t= BV_tB^\top$.

By induction it is possible to prove that $\theta_t$ is the same as $\bar u_t$ in the deterministic case, and $V_t$ does not depend on $x_t$, but solely on $P_{t+1}$. The condition on $\beta$ remains the same one that we found in the previous section.

\paragraph{Step $t=T-1$.} Remember that $J_T^\star(x_T) = -\beta x_T^\top x_T$, then
\begin{align*}
\mathbb{E}[J_T^\star(x_T)|x_{T-1}] &= -\beta \mathbb{E}[(Lx_{T-1}+B\theta_{T-1})^\top(Lx_{T-1}+B\theta_{T-1}) + \trace(\Sigma + R_{T-1})|x_{T-1}].
\end{align*}

Then, at time $T-1$ we have
 \begin{align*}
 J_{T-1}^\star(x_{T-1})&=\min_{\theta_{T-1}, R_{T-1}} \mathbb{E}\Big[\frac{1}{2} \left(\trace(\Sigma^{-1}R_{T-1})+\theta_t^\top B^\top \Sigma^{-1} B\theta_t - \ln|I+\Sigma^{-1}R_{T-1}|\right) \\
 &\qquad -\beta x_{T-1}^\top x_{T-1}+ J_{T}^\star(x_{T}) |x_{T-1}\Big],\\
 &=\min_{\theta_{T-1}, R_{T-1}} \mathbb{E}\Big[\frac{1}{2} \left(\trace(\Sigma^{-1}R_{T-1})+\theta_t^\top B^\top \Sigma^{-1} B\theta_t - \ln|I+\Sigma^{-1}R_{T-1}|\right) \\
 &\qquad -\beta (x_{T-1}^\top (I+L^\top L)x_{T-1} +\trace(\Sigma+R_{T-1})\\
 &\qquad +\theta_{T-1}^\top B^\top B \theta_{T-1} +2x_{T-1}^\top L^\top B \theta_{T-1}) |x_{T-1}\Big].
 \end{align*}
Then, the solution is clearly given  $\theta_{T-1}=\beta B^+F_{T-1}^{-1}L x_{T-1}$, where $F_t$ was defined in the previous section.
We can find the optimal solution for $R_{T-1}$ by solving the equation
\[
\frac{1}{2}\Sigma^{-1} -\beta I -\frac{1}{2} \Sigma^{-1}(I+\Sigma^{-1}R_{T-1})^{-1}=0
\]
Therefore, we can derive the following
\begin{align*}
0&= \frac{1}{2}I- \beta \Sigma - \frac{1}{2}(I+\Sigma^{-1}R_{T-1})^{-1},\\
&= \frac{1}{2}(I+\Sigma^{-1}R_{T-1})- \beta \Sigma(I+\Sigma^{-1}R_{T-1}) - \frac{1}{2} I,\\
&=  \frac{1}{2}(I+\Sigma^{-1}R_{T-1})- \beta\Sigma -\beta R_{T-1} - \frac{1}{2} I.
\end{align*}
Consequently
\[
R_{T-1}\left(\frac{1}{2}\Sigma^{-1}-\beta I\right) =\beta\Sigma.
\]
which implies
$R_{T-1} = \beta \left(\frac{1}{2}\Sigma^{-1}-\beta I\right)^{-1}\Sigma$. Therefore 
\[V_{T-1} = \beta  B^+ \left(\frac{1}{2}\Sigma^{-1}-\beta I\right)^{-1}\Sigma (B^+)^\top.\]
Then
 \begin{align*}
 J_{T-1}^\star(x_{T-1}) &= -\beta x_{T-1}^\top ( I_n + L^\top L +\beta L^\top BF^{-1}B^\top L)x_{T-1} -\beta \trace(\Sigma+R_{T-1})\\ &\qquad+\frac{1}{2}\left[\trace(\Sigma^{-1}R_{T-1}) - \ln|I+\Sigma^{-1}R_{T-1}\right]
 \end{align*}
\paragraph{General solution.} 
 Iterating we can easily find that $\theta_t$ is equal to the solution of the deterministic case, while  the solution of $R_t$ is given by the condition
\[
\frac{1}{2}\Sigma^{-1} -\beta P_{t+1} -\frac{1}{2} \Sigma^{-1}(I+\Sigma^{-1}R_{t})^{-1}=0.
\]
Therefore, using also the symmetry of the matrices, we can derive
\begin{align*}
  &=\frac{1}{2}\Sigma^{-1}(I+\Sigma^{-1}R_{t}) -\beta P_{t+1}(I+\Sigma^{-1}R_{t}) -\frac{1}{2} \Sigma^{-1},\\
 &=\frac{1}{2} \Sigma^{-2}R_{t} -\beta P_{t+1}(I+\Sigma^{-1}R_{t}),\\
 &=\frac{1}{2} R_{t}\Sigma^{-1} -\beta (\Sigma+R_{t})P_{t+1},\\
 &= R_t\left(\frac{1}{2} \Sigma^{-1}-\beta P_{t+1}\right)  -\beta \Sigma P_{t+1}
\end{align*}
Hence,
\[
R_t = \left(\frac{1}{2} \Sigma^{-1}-\beta P_{t+1}\right)^{-1}\beta P_{t+1}\Sigma = \beta F_{t}^{-1}P_{t+1}\Sigma.
\]

  What remains to prove  is to show that
  \[
  -p_t=-\beta \trace((\Sigma+R_{t})P_{t+1})+\frac{1}{2}\left[\trace(\Sigma^{-1}R_{t}) - \ln|I+\Sigma^{-1}R_{t}|\right]\leq 0
  \]
  is non-positive. Using that for square non-singular matrices $A$ we have $\trace \ln A = \ln |A|$ we find
\begin{align*}
	0\geq& -\beta \trace((\Sigma+R_{t})P_{t+1})+\frac{1}{2}\left[\trace(\Sigma^{-1}R_{t}) - \ln|I+\Sigma^{-1}R_{t}|\right],\\
	=& \trace\left(-\beta (\Sigma+R_{t})P_{t+1} +\frac{1}{2}(\Sigma^{-1}R_{t}-\ln (I+\Sigma^{-1}R_{t})\right),\\
	=& \trace\left(R_{t}\left(\frac{1}{2}\Sigma^{-1}-\beta P_{t+1}\right) -\beta \Sigma P_{t+1}-\ln (I+\Sigma^{-1}R_{t})\right)
\end{align*}
Using the fact that $R_t\left(\frac{1}{2} \Sigma^{-1}-\beta P_{t+1}\right)  -\beta \Sigma P_{t+1}=0$ we derive
\begin{align*}
 \trace\left(\beta \Sigma P_{t+1} -\beta \Sigma P_{t+1}-\ln (I+\Sigma^{-1}R_{t})\right)=-\trace\left(\ln (I+\Sigma^{-1}R_{t})\right) \leq 0
\end{align*}
  As required. Therefore, since $P_t$ is the same in both attacks, by comparing $p_t$ one  can  conclude that the value of the problem using a Gaussian attack is lower than the value of a deterministic attack. From the attacker's perspective this implies that a Gaussian attack is better than a deterministic one.
\end{document}